%% file: ms_post_FOCS.tex
\documentclass[11pt]{article}

\input{preamble}

\DeclareMathOperator{\nestwidth}{nw}
\newcommand{\elis}[1]{\normalsize{\color{red}(Elis:\ #1)}}

\makeatletter
\let\latex@fnsymbol\@fnsymbol
\renewcommand\@fnsymbol[1]{\ifcase#1\or*\or1\or2\else\@ctrerr\fi}
\newcommand{\restorefnsymbol}{\let\@fnsymbol\latex@fnsymbol}
\makeatother

\title{Faster shortest-path algorithms \\ using the acyclic-connected tree}

\author{Elis Stefansson\thanks{These authors contributed equally.} \thanks{School of Electrical Engineering and Computer Science, KTH Royal Institute of Technology, Sweden (e-mail: \texttt{\{elisst, kallej\}@kth.se}). The authors are also affiliated with Digital Futures. \newline
This work was supported in part by Swedish Research Council Distinguished Professor Grant 2017-01078, Knut and Alice Wallenberg Foundation Wallenberg Scholar Grant, and Swedish Strategic Research Foundation FUSS SUCCESS Grant.
} 
\and
Oliver Biggar\footnotemark[1] \thanks{Department of Computer Science, Columbia University, United States (e-mail: \texttt{oliver.biggar@columbia.edu}). This work was partially completed while the author was affiliated with CIICADA Lab, School of Engineering, Australian Nation University, Australia.
}
\and 
Karl H. Johansson\footnotemark[2]
}
\date{}

\begin{document}

\maketitle

\begin{abstract}
We provide a method to obtain beyond-worst-case time complexity for any single-source-shortest-path (SSSP) algorithm by exploiting modular structures in graphs. The key novelty is a graph decomposition, called the acyclic-connected (A-C) tree, which breaks up a graph into a recursively nested sequence of strongly connected components in topological order. The A-C tree is optimal in the sense that it maximally decomposes the graph, formalised by a parameter called \emph{nesting width}, measuring the extent to which a graph can be decomposed. We show how to compute the A-C tree in linear time, allowing it to be used as a preprocessing step for SSSP. Indeed, we transform any SSSP algorithm by first computing the A-C tree, and then running the SSSP algorithm in a careful recursive manner on the A-C tree. We illustrate this with two state-of-the-art algorithms: Dijkstra's algorithm and the recent sparse graph algorithm of Duan et al., obtaining improved time complexities of $O(m+n\log(\nestwidth(G)))$ and $O(m\alpha(n)+m\log^{2/3}(\nestwidth(G)))$, respectively, where $\nestwidth(G) \leq n$ is the nesting width of the graph $G$, and $\alpha(n)$ is the extremely slow-growing inverse Ackermann function. Some classes of graphs, such as directed acyclic graphs, have bounded nesting width, and we obtain linear-time SSSP algorithms for these graphs.

\end{abstract}

\section{Introduction}

\subsection{Motivation} \label{sec: motivation}

Finding the shortest paths from a node in a weighted graph is one of the oldest and most fundamental problems in the field of algorithms, whose modern applications range from robotics~\cite{lavalle2006planning} to route planning~\cite{bast2016route}. This problem, called the \emph{Single-Source-Shortest-Path problem} (SSSP), originates from the work of Dijkstra~\cite{Dijkstra1959}, whose eponymous algorithm efficiently computes the tree of shortest paths from a given fixed source node $s$ to all other nodes in a graph $G$. Dijkstra's original algorithm \cite{Dijkstra1959} runs in $O(n^2)$ steps in an $n$-node graph; the addition of the \emph{Fibonacci heap} data structure by Fredman and Tarjan \cite{DijkstraFibonacci} improved the complexity of Dijkstra's algorithm to $O(m+n\log n)$, where $n$ and $m$ are the number of nodes and arcs in the graph, respectively. This matches a known lower bound for Dijkstra's algorithm: because the nodes are visited in sorted order of distance from the source, the algorithm inherits the $\Omega(n\log n)$ lower bound for comparison sorting.





For many decades, it appeared that the SSSP problem was essentially solved (at least in the comparison-addition model \cite{duan2025breaking}). However, two recent breakthroughs (\cite{haeupler2024universal,duan2025breaking}) have called this into question, and the SSSP is now undergoing a significant revival in theoretical computer science. In particular, as noted in \cite{duan2025breaking}, the comparison sorting worst-case lower bound is actually too pessimistic. \textbf{Instead, SSSP can sometimes---and maybe \emph{always}---be solved more efficiently than Dijkstra's algorithm}.

Indeed, to solve SSSP, one only needs to return the shortest-path tree. Dijkstra's algorithm also returns the nodes in sorted order, a \emph{strictly harder} problem, now called \emph{Distance Ordering} (DO). The lower bound for DO, derived from the problem of comparison sorting, does not apply to SSSP. This motivates the work of Duan et al. \cite{duan2025breaking}, who provide the first algorithm to break the $\Omega(n\log n)$ `sorting barrier' for SSSP. Their algorithm returns the shortest-path tree in $O(m\log^{2/3} n)$ steps, which is strictly smaller than $O(m + n \log n)$ for sparse graphs ($m = \Theta(n)$). This result is the first indication that the worst-case complexity of SSSP may be significantly smaller than that of DO. 

In another recent work, Haeupler et al.~\cite{haeupler2024universal} analysed the \emph{beyond-worst-case} performance of Dijkstra's algorithm and found, counterintuitively, that it is not ``universally optimal" (even for the DO problem). Universal optimality measures the worst-case performance of an algorithm on a given input graph, for all possible weightings of the arcs.\footnote{They then show, equipped with a new data structure called a ``working-set heap", that Dijkstra's algorithm can be made universally optimal for DO.} 
These two breakthroughs demonstrate that the complexity of SSSP remains open, and, crucially, that new ideas can lead to significant improvement on this fundamental problem.

Our work extends this line of work, beginning with the following observation: despite their strengths, neither of these new algorithms \cite{haeupler2024universal, duan2025breaking} are \emph{universally optimal for SSSP}. In other words, they do not achieve the worst-case time complexity for SSSP given the graph structure. This is demonstrated by \emph{directed acyclic graphs} (DAGs), where we know that SSSP can be solved in \emph{linear time} \cite{tarjan1972depth}, but both of these state-of-the-art methods take superlinear time (this is clear for \cite{haeupler2024universal} since it is  lower bounded by $\Omega(n\log n)$, and for \cite{duan2025breaking} see Appendix \ref{sec: superlinear}).
\begin{figure}[h]
\centering
\begin{subfigure}{.45\textwidth}
    \centering
    \includestandalone[width=0.8\linewidth]{figs/intro_ex3}
    \caption{}
    \label{fig: intro module in example}
\end{subfigure}
\quad
\begin{subfigure}{.45\textwidth}
    \centering
    \includestandalone[width=0.8\linewidth]{figs/intro_ex3_nested}
    \caption{}
    \label{fig:intro nested}
\end{subfigure}
\caption{A graph $G$, and a subset of its nodes (Fig.~\ref{fig: intro module in example}). This set allows $G$ to be decomposed into a pair of graphs $G_1$ and $G_2$, with $G_2$ nested at a node $X$ of $G_1$ (Fig.~\ref{fig:intro nested}).}
\label{fig: intro example}
\end{figure}

\subsection{Modular structure in graphs} 
Inspired by \cite{haeupler2024universal}, we would like to obtain SSSP algorithms that match the worst-case time complexity of state-of-the-art methods (such as \cite{duan2025breaking}), but outperform these methods when the graph possesses additional structure, with the goal of building towards a universally optimaln algorithm for SSSP.\footnote{Our method is not universally optimal for SSSP, but it is a necessary step in that direction, because any universally optimal SSSP algorithm must perform at least as well as ours on graphs with bounded nesting width. See Section~\ref{sec: conclusions}.} Our focus in this paper is on \emph{modular structure}, that is, \emph{when a graph can be hierarchically decomposed into a nested collection of smaller graphs}. Hierarchical nesting is motivated by the presence of modularity in many real-world applications. It can also be viewed as a generalisation of the DAG case; as we show later, every DAG can be decomposed as a nested collection of two-node graphs. Figure~\ref{fig: intro example} demonstrates such a modular structure in a graph $G$. Here, the set of nodes $S = \{a,d,e,f,g,h\}$ in $G$ has an important property: \textit{all arcs into $S$ go to $h$}. This allows $G$ to be represted as a pair of graphs $G_1$ and $G_2$ (Fig.~\ref{fig:intro nested}), with $G_2$ nested in a node $X$ of $G_1$. Thus, instead of running a black-box SSSP algorithm on $G$, we can instead run it on $G_1$ and then recursively on $G_2$ upon reaching $X$, beginning at its source $h$, with the nonlinear component of the time complexity scaling with the size of $G_1$ and $G_2$. Using Dijkstra's algorithm as an example, decomposing this graph improves the running time from $O(m_1+m_2 + (n_1+n_2)\log(n_1+n_2))$ to $O(m_1 + n_1\log n_1 + m_2 + n_2\log n_2) \leq O(m_1 + m_2 + (n_1+n_2)\log \max(n_1,n_2))$, where $n_i$ and $m_i$ are the number of nodes and arcs in $G_i$, respectively.

\subsection{Contribution}


In this paper, we provide a method to obtain beyond-worst-case time complexity for any SSSP algorithm by exploiting modular structures in graphs. The key idea is to first decompose a given graph into a hierarchy of smaller subgraphs and then run the SSSP algorithm on this decomposition. We first formalise the concept of hierarchical decomposition, by defining a \emph{nesting decomposition} of a graph. In general, a graph typically has several nesting decompositions. We therefore seek one where its \emph{width}, defined as the maximum size of its nested subgraphs, is minimised.\footnote{Intuitively, this corresponds to a graph that is maximally decomposed.} This minimal width, called the \emph{nesting width} (inspired by \emph{treewidth} \cite{robertson1984graph}), is a novel graph parameter that captures the extent to which a graph can be decomposed. We then introduce a novel combinatorial object called the \emph{acyclic-connected} (A-C) tree of a graph, which defines how a graph can be broken into a recursively nested sequence of strongly connected components in topological order. The A-C tree is optimal in the sense that it defines a minimal-width nesting decomposition. Crucially, we give a linear-time algorithm for computing the A-C tree, combining ideas from the linear-time algorithms for strongly connected components \cite{tarjan1972depth} and dominator trees~\cite{alstrup1999dominators}. Finally, by computing the A-C tree as a preprocessing step, we show how we can improve the time complexity of any SSSP algorithm by running it on the A-C tree. In particular, we improve two state-of-the-art algorithms: Dijkstra's algorithm \cite{DijkstraFibonacci} with $O(m+n\log(\nestwidth(G)))$ steps, and the recent sparse graph algorithm \cite{duan2025breaking} with $O(m \alpha(n)+m\log^{2/3}(\nestwidth(G)))$ steps, where $\nestwidth(G) \leq n$ is the nesting width of $G$, and $\alpha(n)$ is the extremely slow-growing inverse Ackermann function \cite{tarjan1984worst}. Many classes of graphs, such as DAGs, have \emph{bounded} nesting width, and so on these graphs we obtain SSSP algorithms that run in linear time, a significant step towards beyond-worst-case optimal SSSP algorithms.

\subsection{Related work} \label{sec: related}
Constructing fast shortest path algorithms has been an active area of research since the `50s, with pioneering work given by \cite{moore1959shortest,shimbel1954structure,bellman1958routing,For56,Dijkstra1959}, culminating in the work of Fredman and Tarjan \cite{DijkstraFibonacci}. Later work focused on strengthening the assumptions on the inputs. One such approach is to constrain the weights, for example by assuming them to be integers \cite{dial1969algorithm,thorup2003integer,bernstein2022negative,bringmann2023negative} or uniformly bounded \cite{dial1969algorithm,ahuja1990faster,johnson1981priority,karlsson1983mlog}. We differ from this line of work since our approach is entirely structural---we assume no restriction on the weights except being non-negative. Other work using structural assumptions include search algorithms for planar graphs \cite{lipton1979generalized,henzinger1997faster,fakcharoenphol2006planar,klein2010shortest,mozes2010shortest}, minor-free graphs \cite{wulff2011separator}, directed acyclic graphs (DAGs) \cite{tarjan1972depth}, and most recently \emph{universal optimality}~\cite{haeupler2024universal}. Generalising the DAG case, \cite{takaoka1998shortest} showed that SSSP can be solved in time $O(m + n\log k)$ where $k$ is the size of the largest connected component. We obtain this as a special case of our method, because 
the size of the largest connected component $k$ is always at least $\nestwidth(G) - 1$, and is often strictly larger. Finally, there are several recent algorithms using continuous optimisation tools \cite{saranurak2019expander,nanongkai2017dynamic,chuzhoy2020deterministic,bernstein2020fully,nanongkai2017dynamic,wulff2017fully}. We work instead directly in a discrete setup to be closer to the original (discrete) problem.


%
%

This paper introduces the A-C tree as a hierarchical decomposition of a graph. Similar concepts include the low diameter decomposition (see \cite{bernstein2022negative,bringmann2023negative} and references within), which groups together strongly connected sets with respect to the weights (essentially, nodes that are close together should be grouped together), with probabilistic guarantees. In particular, \cite{bernstein2022negative,bringmann2023negative} also hierarchically extract acyclic parts using topological order between strongly connected sets to compute shortest path trees faster, but only for integral weights (possibly negative integers). 
Another line of research is the `hierarchy-based' approach developed by \cite{pettie2004new,thorup1999undirected,pettie2005shortest,hagerup2000improved}. This approach divides the graph into subgraphs splitting up the original SSSP into smaller SSSPs to search faster. The basic idea is to find subgraphs that are safe to search in greedily (ignoring the remaining graph) up to a certain distance. This idea is similar to how the A-C tree divides the SSSP into smaller problems. However, these hierarchy-type approaches ultimately rely on the weights, whereas our decomposition is entirely structural. Further, it is shown in \cite{pettie2004new} that any algorithm using this hierarchy method must run in $\Omega(m + n\log n)$ time in the worst case on a certain family of directed graphs. It is straightforward to show that this family has bounded nesting width, and so our algorithm runs in linear time on these graphs.

The A-C tree proposed in this work uses the notion of a \emph{dominator tree}, introduced by \cite{prosser1959applications} and with the first algorithm given by \cite{lowry1969object}. Given a graph with a fixed source node $s$, a node $m$ is said to \emph{dominate} $v$ if every path (from $s$) to $v$ goes through $m$ (Definition~\ref{def: inequality}), and the dominator tree is the tree representing this partial order over all nodes (with root $s$). Efficient algorithms for computing the dominator tree include the classical near-linear time algorithm given by \cite{lengauer1979fast}, and more sophisticated linear time algorithms \cite{alstrup1999dominators,buchsbaum2008linear,buchsbaum1998new,georgiadis2004finding}. In our paper, the dominator tree essentially defines the `acyclic' part of the A-C tree. In this way, the A-C tree can be seen as an extension of the dominator tree, where we also recursively extract strongly connected components. Dominator trees have been used in several applications, such as the design of efficient compilers \cite{aho1977principles,cytron1991efficiently}, modelling ecosystems \cite{allesina2004dominates}, diagnostic test generation \cite{amyeen2001fault}, 2-vertex connectivity tests and approximations \cite{georgiadis2010testing, georgiadis2011approximating}, influencing diffusion networks \cite{gomez2016influence}, finding strong bridges in graphs \cite{italiano2012finding}, diagnosing memory leaks \cite{maxwell2010diagnosing}, and computing reachability in graphs subject to constraints \cite{quesada2006using}. However, dominator trees have rarely been applied to shortest path problems, as we do in this paper. One exception is the recent work of \cite{haeupler2024universal} which we discussed in the introduction. There, they used the dominator tree in a subroutine to simplify the graph as part of the construction of a universally optimal algorithm for distance ordering in terms of the number of comparison queries. Our usage is similar in flavour, though the dominator tree plays a much more central role in our definitions and results.

Our hierarchical approach using modules is inspired by the recent work of \cite{biggar2022modularity,biggar2021modular}. In \cite{biggar2021modular}, the authors study the structure of finite state machines (FSMs) by developing an algorithm to decompose them into \emph{hierarchical finite state machines}. They achieve this by identifying `modules' in the FSM and using this to construct an equivalent nested hierarchy of machines. However, their work did not consider planning on FSMs. This was extended by 
\cite{stefansson2023ecc,stefansson2023cdc}, who showed how to solve SSSP on weighted hierarchical FSMs. Similar to our paper, they exploit the modularity of the hierarchy to apply Dijkstra's algorithm recursively. However, \cite{stefansson2023ecc,stefansson2023cdc} assumes a given hierarchical decomposition of machines, while in this paper we first efficiently construct such a decomposition.

\subsection{Outline}

The remainder of the paper is as follows. Section~\ref{sec:prelims} presents the preliminaries. Section~\ref{sec: AC tree} describes the A-C tree. Section~\ref{sec: correctness} proves that the A-C tree is a minimal-width decomposition. Section~\ref{sec: nesting alg} presents an algorithm for computing the A-C tree and shows that it has linear-time complexity. This algorithm is then used in Section~\ref{sec: rec Dijkstra} to solve SSSP through a recursive application of a given SSSP algorithm on the A-C tree, starting with Dijkstra's algorithm in Section \ref{sec: recursive_dijkstra}, followed by an arbitrary SSSP algorithm in Section \ref{sec: rec SSSP}. An example illustrating the theory is given in Section \ref{sec: example}. Finally, Section~\ref{sec: conclusions} concludes the paper. The appendix presents additional results and proofs.

\section{Preliminaries} \label{sec:prelims}

\subsection{Graphs}

In this paper, a \emph{graph} refers to a weighted directed graph with a distinguished node called the \emph{source}. Specifically, a graph is a tuple $G=(V,E,s)$ where $V$ is a set of \emph{nodes}, $s\in V$ the source, and $E\subset V\times V$ a set of \emph{arcs}. Our graph decomposition uses only this data, but to find shortest paths we will also need a non-negative-valued \emph{weight function} $\ell:E\to\real_{\geq 0}$ which stores the weight for each arc, with graph denoted by $G = (V,E,s,\ell)$. 

We assume that all nodes in $G$ are reachable from $s$. This assumption is mild, because we could easily prune unreachable nodes prior to searching from $s$. Our goal is to solve the \emph{single-source shortest path problem}:

\begin{defn}[Single-source shortest path (SSSP)]
    Given a graph $G = (V,E,s,\ell)$, construct the shortest path tree from $s$ to all other nodes in $V$.
\end{defn}

We solve SSSP using a hierarchical approach. The idea is based on Bellman's principle~\cite{bellman1958routing}, which in this context means that if a node $c$ is on a shortest path $\pth{a}{b}$, then the subpaths $\pth{a}{c}$ and $\pth{c}{b}$ must also be shortest paths. Specifically, if we know that \emph{all} paths from $a$ to $b$ contain $c$, we can reduce the problem of computing the shortest path $\pth{a}{b}$ to the problem of separately finding shortest paths $\pth{a}{c}$ and $\pth{c}{b}$. We formalise this using the dominance relation.

\begin{defn}\cite{lengauer1979fast}\label{def: inequality}
    Let $G = (V,E,s)$ be a graph. If $a$ and $b$ are nodes, we say that \emph{$a$ dominates $b$} (written $a\geq_s b$) if every path $\pth{s}{b}$ contains $a$.
\end{defn}

In the literature on control-flow graphs, this dominance relation is used to compute the \emph{dominator tree} \cite{lowry1969object}, since the relation defines a tree order on the nodes with root~$s$. 

\begin{lem}[\cite{lengauer1979fast}] \label{tree order lemma} 
    The relation $\geq_s$ is a \emph{tree-structured partial order} on the nodes of $G$, rooted at $s$. That is, (1) $\geq_s$ is a partial order, (2) for all $a\in V$, $s\geq_s a$ and (3) if $a\geq_s c$ and $b\geq_s c$, then either $a\geq_s b$ or $b\geq_s a$.
\end{lem}

\begin{defn}[Dominator tree \cite{lengauer1979fast}]
    The dominator tree of a graph $G = (V,E,s)$ is the transitive reduction of the ordering $\geq_s$.
\end{defn}

The concept of dominance is closely related to---but not exactly the same as---the sets of nodes which can be treated as a recursively nested graph, as in Section~\ref{sec: motivation}. We formalise such sets through the notion of a \emph{module}.



\begin{defn}[Module] \label{def: module}
    A \emph{module} of $G = (V,E,s)$ is a set $M\subseteq V$ of nodes such that there exists an $m \in M$ where for every arc $\arc{u}{v}$, if $u\not\in M$ and $v\in M$ then $v=m$. In other words, all arcs into $M$ go to $m$. 
    We call $m$ the \emph{source} of $M$.\footnote{Note that `modules' also have a distinct interpretation in graph theory \cite{gallai1967transitiv,habib2010survey}. Here we use a form specific to single-source graphs, which is closer to the usage in \cite{biggar2021modular,biggar2022modularity}.}
\end{defn}
We also use the convention that if a set $M\subseteq V$ contains the source $s$ of $G$, then $M$ is a module \emph{if and only if} $M$ is a module with source $s$. Intuitively, one can think of an extra arc into $s$ from outside $G$, capturing that $s$ is the source of $G$.

Because we assume all nodes are reachable from $s$, every set that is not the whole node set $V$ must have at least one incoming arc. Also note that if $M$ is a module with source $m$, then $m\geq_s v$ for every $v\in M$.  Every graph has at least some modules: the whole set $V$ is always trivially a module, as are the singleton sets $\{v\}$ for each $v\in V$. We call these the \emph{trivial modules}. 

Modules have some nice algebraic properties. In particular, the following will be useful.

\begin{lem} \label{lem: intersection}
    If $M$ and $H$ are overlapping modules (they intersect but neither contains the other), then $M\cup H$ and $M\cap H$ are both modules.
\end{lem}
\begin{proof}
    See Appendix \ref{sec: proof sec 2}.
\end{proof}

Finally, as discussed in the introduction, SSSP can be solved in linear time for DAGs, using the following simple method. First, one computes a topological ordering of the nodes and then visits the nodes in this order, updating distances to successor nodes as we visit each node \cite{tarjan1972depth}. This subroutine will play a key role in our algorithms since we work with nested acyclic graphs.

\subsection{Nesting width} \label{sec: nesting_width}

In Figure~\ref{fig: intro example}, we could treat the set $S$ as a nested graph $G_2$ with source node $h$ because all arcs into $S$ goes to $h$. In other words, we could represent $S$ as a nested graph precisely because it is a module. Note that if $S$ were not a module, say because an arc $\arc{b}{d}$ existed in Figure~\ref{fig: intro module in example}, then we would not be able to assign a unique source node to $G_2$. Hence the sets of modules define the \emph{nesting decompositions} of a graph.

\begin{defn}[Nesting decomposition] \label{def: nesting}
    Let $G = (V,E,s)$ be a graph. A \emph{nesting decomposition} of $G$ is a collection of non-overlapping modules which includes all trivial modules (the whole set $V$, and the singletons $\{v\}$ for each $v\in V$).
\end{defn}
Each module $M$ in a nesting decomposition $F$ has a unique \emph{maximal module partition}, which consists of the modules $H \in F$ such that $H\subset M$ and there is no module $K \in F$ with $H\subset K \subset M$. This is a necessarily a partition because we include the trivial singleton modules, and it is unique because all modules in the set are non-overlapping. The \emph{width} of a nesting decomposition is the largest size of the maximal module partition for any module in the collection.

\begin{defn}[Nesting width] \label{def: nesting width}
    The \emph{width} of a given nesting decomposition is the maximum size of a maximal module partition of any module in the decomposition. The \emph{nesting width} of a graph $G = (V,E,s)$, written $\nestwidth(G)$, is the minimum width over all nesting decompositions of $G$.
\end{defn}

A nesting decomposition $F$ of $G$ can be equivalently represented by a tree. Indeed, consider the graph whose nodes are the modules in the decomposition $F$, with an arc $\arc{M}{H}$ if $M\supset H$ and no $K \in F$ such that $M\supset K \supset H$. This graph defines a tree whose leaves are the singleton modules. In this form, the \emph{width} of the decomposition is the maximum number of children of any node in this tree.




Note that every graph has at least one trivial nesting decomposition, consisting of only the trivial modules $V$ and $\{v\}$ for each $v\in V$. As a tree, this is a single root node with all singletons as its children. In particular, we always have the upper bound $\nestwidth(G) \leq n$. Generally, though, graphs can have many distinct nesting decompositions. We want to find the `best' nesting decomposition, whose width is minimal and thus equal to the nesting~width. Such a decomposition is given in the next section. 

We finish this section by stating how the nesting width generalises DAGs given by the following proposition.

\begin{prop}[Nesting width and DAGs]\label{prop: generalise_dags}
For any graph,\footnote{For $G$ with more than one node. When $G$ has just one node, then $\nestwidth(G)=1$.} $2 \leq \nestwidth(G) \leq n$. In particular, for a DAG $G$, the lower bound $\nestwidth(G)=2$ is attained. In this sense, $\nestwidth(G)$ can be seen as a measure of how acyclic $G$ is.
\end{prop}

\begin{proof}
See Appendix \ref{sec: proof sec 2}.
\end{proof}

\begin{rmk}[Nesting width and sparsity]
Sparse graphs are more likely to have small nesting width, but sparse graphs and bounded nesting width graphs are distinct concepts. For details, see Appendix \ref{sec: nesting width vs sparsity}.
\end{rmk}

\section{The Acyclic-Connected Tree} \label{sec: AC tree}

Nesting decompositions are defined by modules, which are themselves related to the dominance relation $\geq_s$. This is the first ingredient of our decomposition: using the dominator tree to compute the relation $\geq_s$. However, dominance alone is not enough to construct a minimal-width nesting decomposition. Consider the following example: a graph with nodes $s, a_i, b_i$ for $i$ from 1 to $N$, where there are arcs from $s$ to $a_1$ and $b_1$, and arcs from $a_i$ and $b_i$ to both $a_{i+1}$ and $b_{i+1}$. The dominator tree of this graph is the trivial one-layer tree with root $s$. However, this is acyclic, so a nesting decomposition of width $2$ exists (follows from Proposition \ref{prop: generalise_dags}). More generally, while the dominator tree tells us that all modules have source node $s$, this is not sufficient to identify them.

In this example we can construct a minimal-width decomposition by topological sorting of the graph, given it is acyclic. Our main theorem proves that these two ideas---dominators and topological sorting---are \emph{always sufficient} to construct a minimal-width decomposition of any graph. Collectively, we will use these to define the \emph{acyclic-connected tree} (A-C tree) of the graph.

\begin{defn}[A-C tree] \label{def: AC tree}
    Let $G = (V,E,s)$ be a graph, and $T$ be its dominator tree. For any node $a$ in $T$, let $D(a)$ denote its \emph{descendants} and $C(a)$ its \emph{children}. For each $a$, we define a graph $G_a$ whose node set is $C(a)$, with an arc $\arc{u}{v}$ in $G_a$ if $u \neq v$ and there is an arc from $D(u)$ to $D(v)$ in $G[D(a)]$ (the induced subgraph of $G$ consisting of descendants of $a$). The \emph{acyclic-connected tree} (A-C tree) of $G$ is a map $a\mapsto S^a_1,\dots,S^a_{k_a}$ taking each node $a$ to the sequence of strongly connected components of $G_a$ in topological order.\footnote{There may be many topological orders, so this is not unique. However, the choice of ordering will not affect any analysis, so we will treat it as unique for simplicity.} We depict this as in Figure~\ref{fig:ac_decomposition}.
\end{defn}
\begin{figure}[h]
\centering
\begin{subfigure}{.4\textwidth}
    \centering
    \includestandalone[width=0.8\linewidth]{figs/intro_plain}
    
    \caption{}
    \label{fig:saved_and_sources}
\end{subfigure}
\quad
\quad
\begin{subfigure}{.4\textwidth}
    \centering
    \includestandalone[scale=.8]{figs/AC_tree_new}
    
    \caption{}
    \label{fig:AC_tree}
\end{subfigure}
\caption{
The graph $G$ from Figure~\ref{fig: intro example} and its associated A-C tree (Fig.~\ref{fig:AC_tree}), where the children of a node are topologically ordered from left to right. The child of each node $x$ in the tree is a set of nodes which represent a strongly connected component of $G_x$. We represent this by grouping the nodes in a rectangle (e.g., $e$ and $d$). We omit this rectangle for clarity if the strongly connected component is a singleton.}
\label{fig:ac_decomposition}
\end{figure} 


In the next section, we prove that the A-C tree defines a minimal-width nesting decomposition of the graph, and show how the A-C tree can be computed in linear time.

\section{A-C trees are minimal-width nesting decompositions} \label{sec: correctness}

In this section we explain why the A-C tree gives us an minimum-width nesting decomposition of a graph. Formally, we have the following result.

\begin{thm} \label{thm: correctness}
    Let $T_{AC}$ be the A-C tree of a graph $G = (V,E,s)$, so $T_{AC}$ maps each node $a$ to a sequence of sets $S^a_1,\dots,S^a_{k_a}$. Then the family of sets $AC:= \{\{a\}\cup\bigcup_{i=1}^{j} D(S^a_i) : a\in V, j\in \{1,\dots,k_a\}\}$ is a minimum-width nesting decomposition of $G$, where $D(S_i^a)$ is the set of descendants of the set $S_i^a$ in the A-C tree.
\end{thm}
\begin{proof}
    We begin by showing that $AC$ indeed defines a nesting decomposition, and its width is equal to $1 + \max_{a,i} |S^a_i|$. To do this we must show that every set is a module, and that no two sets overlap.

    Fix $a$, and let $S_1^a,S_2^a,\dots,S_{k_a}^a$ be the associated sequence of strongly connected components of $G_a$. 
    The sets $\{a\}\cup\bigcup_{i=1}^{j} D(S^a_i)$ for $j\in \{1,\dots,k_a\}$ form a nested increasing sequence, with $\{a\}\cup\bigcup_{i=1}^{j} D(S_i^a) \subset \{a\}\cup\bigcup_{i=1}^{k} D(S_i^a)$ for $j<k_a$. The singleton $\{a\}$ is a module, and all nodes in $D(S_i^a)$ are dominated by $a$. Because the strongly connected components $S_i^a$ are in topological order with respect to $G_a$, any arcs into $D(S_j^a)$ must originate in $\{a\}\cup\bigcup_{i=1}^{j-1} D(S_i^a)$. Hence, by induction, each $\{a\}\cup\bigcup_{i=1}^{j} D(S_i^a)$ is a module. This gives us a collection of non-overlapping sets for every $a$. 

    Finally suppose that a node $v$ is in two sets $M$ and $H$ in $AC$, which have distinct source nodes $a$ and $b$. Then $v$ is dominated by both $a$ and $b$, and so by Lemma~\ref{tree order lemma} either $a\geq_s b$ or $b\geq_s a$. Without loss of generality, we assume the former. In particular, $b$ is a descendant of some child $S_i^a$ of $a$. But then all descendants of $b$ are also in $D(S_i^a)$, and so we have that $H\subseteq M$. Hence $AC$ is a nesting decomposition of $G$.

    Pick a module $M := \{a\}\cup\bigcup_{i=1}^{j} D(S^a_i)$ in $AC$. We can precisely identify the size of its maximal module partition. First, the set $\{a\}\cup\bigcup_{i=1}^{j-1} D(S^a_i)$ is a module in $AC$, and a subset of $M$, and by construction it is maximal. Hence, the remaining modules in the maximal module partition must be subsets of $D(S_j^a)$ with sources other than $a$. The nodes $v$ in $S_j^a$ are the children of $a$ in the dominator tree. For each such $v$, the set $D(v) = \{v\}\cup\bigcup_{i=1}^{k_v} D(S^v_i)$ is a module in $AC$. Because $v$ is a child of $a$ in the dominator tree, $v$ cannot be dominated by any other elements of $S_j^a$. In particular this means that the module $D(v)$ must be maximal inside $M$. Hence the maximal module partition of $M$ in $AC$ is exactly 
    \begin{equation} \label{module set}
        \left \{ \{a\}\cup\bigcup_{i=1}^{j-1} D(S^a_i),D(v_1),\dots,D(v_m)
    \right \}
    \end{equation}
    where $v_1,v_2,\dots,v_m$ are the nodes in $S^a_j$. Hence, the size of the maximal module partition of $M$ is equal to $1 + m = 1 + |S^a_j|$. The width of $AC$ is the size of the largest maximal module partition of any module in $AC$, which is $1 + \max_{a,i} |S^a_i|$. We denote this quantity by $w_{AC}$. 
    
    We will now show that no other nesting decomposition can have strictly smaller width than $w_{AC}$. For contradiction, suppose that a nesting decomposition $T$ exists with width $w < w_{AC}$. Note that every decomposition has width at least two, so this implies that $w_{AC}$ is at least three.

    Since $w_{AC} = 1 + \max_{a,i} |S^a_i|$, we can choose some $\alpha$ and $k$ such that $|S^\alpha_k| = w_{AC} - 1$. Let $K$ be the smallest module in $T$ that contains all of $S^\alpha_k$. By assumption, the maximal module partition of $K$ in $T$ has size at most $w < w_{AC}$. Therefore, there must exist a module $M$ in the maximal module partition of $K$ that contains a subset of $S^\alpha_k$ but not the whole set. In particular, we have two nodes $v_1 \in S^\alpha_k \setminus M$ and $v_2 \in S^\alpha_k \cap M$. Now, because $S^\alpha_k$ is strongly connected (and the definition of $G_\alpha$), there is a path $\pth{v_1}{v_2}$ in $D(S^\alpha_k)$. For $M$ to be a module, its source $\beta$ must be contained in $\pth{v_1}{v_2}$ (since we go from $v_1 \notin M$ to $v_2 \in M$). Moreover, by definition of the dominator tree, there is also a path $\pth{s}{v_2}$ that does not contain any of the nodes of $\pth{v_1}{v_2}$ except $v_2$, and since $\beta$ must also be contained in $\pth{s}{v_2}$,\footnote{This follows since $s \notin K$ but $v_2 \in K$. To see the former, note that if $s \in K$, then $\beta  = s$. But $\beta$ is contained in the path $\pth{v_1}{v_2}$ in $D(S^\alpha_k)$, so $s \in D(S^\alpha_k)$, which is a contradiction since $|S^\alpha_k| \geq 2$.} we get that $\beta = v_2$. Therefore, $M \subseteq D(v_2)$. Since $M$ was arbitrary and the maximal module partition of $K$ must cover $S^\alpha_k$, we get that the number of modules in this partition that overlaps $S^\alpha_k$ must be at least $|S^\alpha_k|$ and these modules are all contained in $D(S^\alpha_k)$. However, $K$ must also have nodes outside of $D(S^\alpha_k)$.\footnote{To see it, assume $K$ is contained in $D(S^\alpha_k)$. Since $K$ contains $S^\alpha_k$, the source $\beta$ of $K$ must belong to $S^\alpha_k$. However, since $|S^\alpha_k| \geq 2$, we may pick another node $v \neq \beta$ in $S^\alpha_k$. By definition of the dominator tree, there exists a path $\pth{s}{v}$ that does not contain $\beta$, a contradiction since $v \in K$.} Therefore, there exist at least one other module $M'$ in the maximal module partition of $K$ that accounts for the nodes outside of $D(S^\alpha_k)$. We conclude that the maximal module partition of $K$ must have size at least $|S^\alpha_k|+1 = w_{AC}$, so $w \geq w_{AC}$, a contradiction.
\end{proof}

\section{Computing the A-C tree} \label{sec: nesting alg}

We proved in Theorem~\ref{thm: correctness} that the A-C tree defines a nesting decomposition of minimal width. However, for this to be useful for solving SSSP, we must also be able to compute the A-C tree efficiently. In this section we combine some classical linear-time algorithms (dominator trees, depth-first-search and strongly connected components) to achieve a linear-time algorithm for computing the A-C tree (Algorithm~\ref{alg:ac decomposition}).

The first (and most technical) step of A-C tree computation is the construction of the dominator tree. Fortunately, several decades of algorithmic work has established linear-time algorithms for this problem \cite{alstrup1999dominators}. The second step, recalling Definition~\ref{def: AC tree}, is to group the children $C(a)$ of a node $a$ in the dominator tree into components, defined by reachability in the induced subgraph $D(a)$. The main challenge of this step is to efficiently construct the induced graphs $G_a$ whose nodes are $C(a)$ and where there is an arc $\arc{u}{v}$ in $G_a$ if there is a path $\pth{u}{v}$ in $G[D(a)]$. In Algorithm~\ref{alg:compress}, we show how to construct the graphs $G_a$ for every $a$ in linear time using a variant of depth-first search. Once each graph $G_a$ has been constructed, we can run a linear-time strongly connected components algorithm (such as \cite{tarjan1972depth}) on each $G_a$ from $a$ to generate the components of $G_a$ in topological order, giving the children of $a$ in the A-C tree.

\begin{algorithm}
\SetAlgoLined
\SetKwInOut{Input}{Input}
\SetKwInOut{Output}{Output}
\SetKwProg{Init}{A-C\_decomposition}{($G,source$)}{}
\SetKwProg{nestdecomp}{A-C\_recursive}{}{}
\SetKw{nestdecompfunc}{A-C\_recursive}
\SetKw{acyclicdecomp}{A\_decomposition}
\SetKw{merge}{merge}
\SetKwData{opened}{is\_opened}
\SetKwData{closed}{is\_closed}
\SetKwData{pred}{predecessors}
\SetKwData{count}{count}
\SetKwData{true}{True}
\SetKwData{false}{False}
\SetKwData{none}{None}
\SetKwData{initcount}{initial\_count}
\SetKwData{indegree}{indegree}
\SetKwData{merged}{merged}
\SetKwData{find}{find}
\SetKwData{flag}{acyclic\_root}
\SetKwData{group}{group\_parent}
\SetKw{comp}{compress}
\Input{Graph $G = (V,E,s)$}
\Output{A-C tree $T_{AC}$}
Create an empty dictionary $T_{AC}$\;
Construct the dominator tree $T$ of $G$\;
$graphs \gets \texttt{DominanceGraphs}{(G,s,T)}$\; 
\For{each node $a$}{
    $K_1,K_2,\dots,K_m \gets SCCs(graphs[a])$\tcp*{Topologically sort $G_a$} 
    Set $T_{AC}[a] \gets K_1,K_2,\dots,K_m$\;
}
\Return $T_{AC}$\;

 \caption{\texttt{A-C\_tree}}
 \label{alg:ac decomposition}
\end{algorithm}

\begin{algorithm}
    \SetAlgoLined
\SetKwInOut{Input}{Input}
\SetKwInOut{Output}{Output}
\SetKwData{found}{found}
\SetKwData{parent}{parent}
\SetKwData{current}{current}
\SetKwData{graph}{graph}
\SetKwProg{compress}{\texttt{RecursiveDominanceGraphs}}{}{}
\SetKw{comp}{\texttt{RecursiveDominanceGraphs}}
\Input{Graph $G$, with source $s$ and dominator tree $T$}
\Output{$G_p$ for each node $p$ given by $graphs$}
Run \comp{}($s$)\;
$graphs[p] \gets p.\graph$ for each node $p$\;
\Return $graphs$\;
\compress{$(v)$}{
    $v.\found  \gets True$\;

    \For{each arc $(v,w)$ in $T$}{
        \If{$w.\found = False$}{
            $v.\current \gets w$\;
            \comp{$(w)$}
        }
    }
    \For{each arc $(v,w)$ in $G\setminus T$}{
        $p \gets w.\parent$\;
        \If{$p.\current \neq w$ \label{lineref: check}}{
        add the arc $(p.\current,w)$ to $p.\graph$\; \label{lineref: adding arc}
        }
    }
}

\caption{\texttt{DominanceGraphs}}
\label{alg:compress}
\end{algorithm}

\begin{thm}
Given a graph $G=(V,E,s)$, Algorithm~\ref{alg:ac decomposition} constructs its A-C tree in time~$O(n+m)$.
\end{thm}
\begin{proof}
Algorithm~\ref{alg:ac decomposition} begins by constructing the dominator tree $T$ of $G$, which can be done in linear time \cite{alstrup1999dominators}. Then we run the subroutine \texttt{DominanceGraphs$(G,s,T)$} (Algorithm~\ref{alg:compress}).

(\emph{Claim:} Algorithm~\ref{alg:compress} computes $G_a$ for every $a$.) Algorithm~\ref{alg:compress} performs a depth-first search (DFS) on the graph $G$ starting at $s$. However, the key property is that the specific DFS tree it follows is the dominator tree $T$. It is immediate that this algorithm is linear time, as it examines each node and arc exactly~once.

Fix a node $a$. We would like to show that $a.graph$ equals $G_a$ after running  Algorithm~\ref{alg:compress}. To do this, we must show that for every distinct $x,y\in C(a)$ where there is an arc from $D(x)$ to $D(y)$ in $G[D(a)]$, we add the arc $\arc{x}{y}$ to $a.graph$, and conversely, every added arc $\arc{x}{y}$ to $a.graph$ must be an on this form. First, observe that at the time we call $\texttt{RecursiveDominanceGraphs}(a)$, no descendants of $a$ have yet been visited. By running a DFS from $a$ in the dominator tree, $\texttt{RecursiveDominanceGraphs}(a)$ visits every descendant of $a$ exactly once. For each $v\in D(a)$, we examine each outgoing arc $\arc{v}{w}$ exactly once.

To this end, suppose first there is an arc $\arc{v}{w}$ in $G[D(a)]$ where $v\in D(x)$ and $w\in D(y)$ for distinct $x,y\in C(a)$. Then, while running $\texttt{RecursiveDominanceGraphs}(x)$, we reach the node $v$ and examine the arc $\arc{v}{w}$. Because $T$ is a dominator tree, $D(y)$ is a module with source $y$. Thus, any arc into $D(y)$ from outside $D(y)$ must go to $y$. Since $v \notin D(y)$, the arc $\arc{v}{w}$ is such an arc, and therefore $w = y$. Then $w.parent = y.parent = a$. Because we are currently running $\texttt{RecursiveDominanceGraphs}(x)$, $a.current = x$. Hence on line~\ref{lineref: adding arc} we add the arc $\arc{x}{y}$ to $a.graph$. 

Conversely, assume we add an arc $\arc{x}{y}$ to $a.graph$ on line \ref{lineref: adding arc}, while examining some arc $\arc{v}{y}$. We want to show that $\arc{x}{y}$ is then an arc of $G_a$. Note that $y.parent = w.parent =a$, so $y \in C(a)$. Moreover, $x = a.current$ and thus, $x \in C(a)$. Since $x$ and $y$ are distinct by line~\ref{lineref: check} (meaning $v$ is not dominated by $y$), it remains to show that $v \in D(x)$. To see this, observe that $v \in D(a)$, since the arc $\arc{v}{w}$ goes to $w \in D(a)$ and $w \neq a$. Because $v$ is the current node in the DFS and $x$ is on the current DFS path from $a$, we get that $v \in D(x)$. Thus, $\arc{x}{y}$ is an arc of $G_a$.

We conclude that for each $a$, Algorithm~\ref{alg:compress} results in $a.graph$ being equal to $G_a$.


Finally, we iterate through each node $a$ and run Tarjan's strongly connected components algorithm \cite{tarjan1972depth} on $graph[a]$, which we know is equal to $G_a$. This algorithm runs in linear time, and returns the strongly connected components $K_1,K_2,\dots,K_m$ of $G_a$ in topological order. Note that because of the tree structure, each node is topologically sorted exactly once, ensuring the overall procedure is still linear time. The resultant map $T_{AC}$ is the A-C tree, by Definition~\ref{def: AC tree}.
\end{proof}

\section{Shortest paths from the A-C tree} \label{sec: rec Dijkstra}

From the A-C tree, we implicitly obtain a minimum-width nesting decomposition. In this section, we show how this can be exploited to improve state-of-the-art shortest-path~algorithms. 

First (Section~\ref{sec: recursive_dijkstra}), we show how, using the A-C tree, we can modify just a few lines of Dijkstra's algorithm and obtain an algorithm (Algorithm~\ref{alg:recursive_dijkstra_main}) with improved time complexity $O(m + n\log \nestwidth(G))$. To the best of our knowledge, this is currently the best asymptotic improvement of Dijkstra's algorithm for general graphs. In particular, for classes of graphs with bounded nesting width, Algorithm~\ref{alg:recursive_dijkstra_main} runs in linear time.

Second (Section~\ref{sec: rec SSSP}), we show how to apply the A-C tree to any \emph{any} black-box SSSP algorithm (such as the recent work of \cite{duan2025breaking}). More precisely, we present an algorithm called Recursive SSSP (Algorithm~\ref{alg:rec sssp2}) that takes as input a black-box SSSP algorithm $\mathcal{SP}$ and applies it recursively using the A-C tree. Given a time complexity bound $f_{\mathcal{SP}}(n,m)$ for $\mathcal{SP}$, the time complexity of Recursive SSSP is $O(m \alpha(n) + \sum_{K_i} f_{\mathcal{SP}}(n_i,m_i))$ (Theorem~\ref{thm: recursive_SSSP}), where the sum is taken over all components $K_i$ of the A-C tree, $n_i$ ($m_i$) is the number of internal nodes (arcs) in component $K_i$, and $\alpha(n)$ is the extremely slow-growing inverse Ackermann function. This results in improved time complexity under mild conditions. Indeed, whenever $f_{\mathcal{SP}}(n,m) \in \Omega(m \alpha(n))$ 
this is an asymptotic improvement. To illustrate this, we apply Recursive SSSP on the recent sparse graph algorithm of \cite{duan2025breaking}, which runs in time $O(m\log^{2/3}(n))$. The result is an algorithm with improved time complexity $O(m \alpha(n)+m\log^{2/3}(\nestwidth(G)))$ (Corollary \ref{corol: sparse_graph_algorithm}). In particular, for sparse graphs $m=O(n)$, our algorithm runs in $O(n \alpha(n)+n\log^{2/3}(\nestwidth(G)))$. To the best of our knowledge, this is best beyond-worst-case complexity bound for SSSP on sparse graphs.

\subsection{Recursive Dijkstra} \label{sec: recursive_dijkstra}

We will begin with Algorithm~\ref{alg:recursive_dijkstra_main}, which we call Recursive Dijkstra. The idea is simply that, upon visiting a node $a$, we first recursively search its children $K^a_1,K^a_2,\dots,K^a_m$ in the A-C tree, where again we search recursively when we reach new nodes. Each component $K^a_i$ contains its own priority queue, which is updated individually.

\begin{thm}\label{thm:recursive_dijkstra}
    Recursive Dijkstra solves SSSP in $O(m+n\log\nestwidth(G))$ time.
\end{thm}

\begin{algorithm}[t]
    \SetAlgoLined
    \SetKwProg{main}{RecursiveDijkstra}{}{}
    \SetKw{mainfunc}{RecursiveDijkstra}{}{}
    \SetKwProg{Asearch}{Asearch}{}{}
    \SetKwProg{Csearch}{ComponentDijkstra}{}{}
    \SetKwProg{Updist}{Updist}{}{}
    \SetKw{recdijkstrafunc}{recursive\_step}
    \SetKw{Asearchfunc}{Asearch}
    \SetKw{Csearchfunc}{ComponentDijkstra}
    \SetKw{Updistfunc}{Updist}{}{}
    \SetKwProg{Init}{Main}{}{}
    \SetKwData{arcorder}{arc\_order}
    \SetKwData{indegree}{indegree}
    \SetKwData{count}{count}
    \SetKwData{dist}{dist}
    \SetKwData{parent}{parent}
    \SetKwData{weight}{weight}
    \SetKwData{queue}{queue}
    \SetKwData{component}{component}
    \SetKwInOut{Input}{Input}
    \SetKwInOut{Output}{Output}
    \Input{Graph $G = (V,E,s)$}
    \Output{Shortest distances and path tree given by \dist and \parent}
        Create two dictionaries $\dist$ and $\parent$\;
        For each $v$, set $\parent[v] \gets None$ and $\dist[v] = \infty$\;
        Set $\dist[s]\gets 0$\;
        {\color{Blue}
        Compute $T_{AC} \gets \texttt{A-C\_tree}(G)$\;
        Assign a priority queue to each component $K^a_i$ with the nodes in $K^a_i$\;}
        Run \mainfunc{$(s)$}\;
        \Return \dist, \parent\;
    {\color{Blue}
    \main{$(v)$}{
        
        $K_1,K_2,\dots,K_m \gets T_{AC}[v]$\;
        \For{each $i$ from 1 to $m$}{
            \Csearchfunc{$(K_i)$}\;
        }
    }}
    \Csearch{$(\kappa)$}{
        $Q\gets \kappa.\queue$\;
        \While{$Q$ is not empty}{
            $(dist,v) \gets Q.extract\_min()$\;
            $\dist[v] \gets dist$\;
            {\color{Blue} \mainfunc{$(v)$}\;} 
            \For{each successor $w$ of $v$}{
                    $L \gets \dist[v] + \weight(v,w)$ \tcp*{Potentially new distance to $w$}
                    \If{$L < \dist[w]$}{
                        $\dist[w] \gets L$ \tcp*{Update distance to $w$}
                        $\parent[w] \gets v$ \tcp*{Save new shortest path to $w$}
                        {\color{Blue} $K \gets$ component $K^a_i$ containing $w$ in $T_{AC}$\;
                        $K.\queue.decrease\_key(w,L)$ \tcp*{Update priority queue}} \label{lineref: update queue}
                    }
            }
        }
    }
    \caption{\texttt{Recursive Dijkstra}} 
    \label{alg:recursive_dijkstra_main}
\end{algorithm}

\begin{proof}

\textbf{Correctness.}
Overall, Recursive Dijkstra runs very similarly to the standard setup of Dijkstra's algorithm, with two key differences: (1) for each node $a$ in the A-C tree it searches each connected component $K^a_i$ separately in topological order, and (2) it applies this process recursively using the A-C tree structure. In Algorithm~\ref{alg:recursive_dijkstra_main}, we highlight the differences from Dijkstra's algorithm in blue. We will explain these below.

First, we show that solving SSSP separately over strongly connected components works. Suppose the dominator tree is trivial, with $s$ dominating all nodes but no other dominance relations. Then this algorithm reduces to solving SSSP by running Dijkstra's separately on each connected component $K_1,K_2,\dots,K_n$ in topological order. We will first show that this produces the same output as running Dijkstra's algorithm on the entire graph collectively. The shortest paths to nodes in $K_1$ does not depend on any other $K_i$, because no paths exist from $K_i$ to $K_1$. Hence Dijkstra's algorithm on $K_1$ alone successfully finds the shortest paths to all nodes in $K_1$. Then we run it on $K_2$, having previously found all shortest paths to $K_1$. But this operates the same as running Dijkstra's algorithm on $K_1\cup K_2$, where we already have solved shortest paths to $K_1$. Again, the larger $K_i$ do not affect the shortest paths, so again this correctly identifies shortest paths to $K_1\cup K_2$. The remainder of the argument proceeds identically, by induction.

Now we argue that this works in the recursive setting as well. Suppose that we are searching from a node $a$ in $T_{AC}$. Assume by induction that until this point, running Recursive Dijkstra on any node $x$ correctly identifies the shortest paths from $x$ to all nodes in $D(x)$. We will show that this holds as well for $a$. Let $K^a_1,K^a_2,\dots,K^a_m$ be the components of $a$, which again we will search in topological order. The only difference from the base case is that an arc $\arc{x}{y}$ in $G_a$ does not necessarily correspond to an arc $\arc{x}{y}$ in $G$---instead, it might correspond to an arc $\arc{u}{w}$ where $u\in D(x)$ and $w\in D(y)$. Suppose this is the case, for some $x$ and $y$ with $y\in K^a_i$. First, recall that because $D(y)$ is a module, $w = y$. Then, by the inductive hypothesis, we have already found the shortest path from $x$ to $u$, because $u\in D(x)$. Hence we can update the distance to $y$ in $K^a_i.queue$ by using $dist[u] + weight(u,y)$ (line~\ref{lineref: update queue}). This operation takes place while searching recursively from $u$. That is, once we have run Recursive Dijkstra on $x$, we have correctly updated the distances to all successors $y$ in $G_a$. Hence searching the components in topological order in $G_a$ works as before, allowing us to correctly compute shortest paths from $a$ to all nodes in $D(a)$.

\textbf{Time complexity.} Having proved the correctness, it remains to show the time complexity of Algorithm \ref{alg:recursive_dijkstra_main}. To this end, note first that all subroutines in Algorithm \ref{alg:recursive_dijkstra_main} are done in linear time, except running Dijkstra's algorithm (given by \textbf{ComponentDijkstra}) on each strongly connected component. We compute this time complexity. Totally analogous to Section \ref{sec: motivation}, assuming $m$ strongly connected components $K_1,\dots,K_m$ in the A-C tree, the time complexity to run Dijkstra's algorithm on all these components is $O\big ((\sum_{i=1}^m m_i) + (\sum_{i=1}^m n_i)\log (\max_i n_i) \big )$, where $n_i$ and $m_i$ is the number of nodes and edges in $K_i$. We note that $\sum_{i=1}^m n_i = n$ and $\sum_{i=1}^m m_i \leq m$, where the latter holds since some arcs may not be in the components. Moreover, $\max_i n_i = \max_i |K_i| = w_{AC}-1=\nestwidth(G)-1 \leq \nestwidth(G)$ by the proof of Theorem~\ref{thm: correctness}. We conclude that the time complexity to run Dijkstra's algorithm on all components is
\begin{equation*}
O \left ( \left (\sum_{i=1}^m m_i \right ) + \left (\sum_{i=1}^m n_i \right )\log \left (\max_i n_i \right ) \right ) \leq O \left (m+n\log\nestwidth(G) \right ).
\end{equation*}
This is also the time complexity of Algorithm \ref{alg:recursive_dijkstra_main}, as all other subroutines are executed in linear time.
\end{proof}
Theorem \ref{thm:recursive_dijkstra} shows that Recursive Dijkstra is asymptotically faster than Dijkstra's algorithm (since $\nestwidth(G) \leq n$). In particular, for graphs of bounded nesting width, Recursive Dijkstra runs in linear time:
\begin{corol}\label{corol:bounded_nesting_width_linear_SSSP}
If a class of graphs has bounded nesting width, then 
Recursive Dijkstra runs in linear time $O(n+m)$ on this class.
\end{corol}

\begin{algorithm}[ht!]
    \SetAlgoLined
    \SetKwProg{main}{TreeRecursion}{}{}
    \SetKw{mainfunc}{TreeRecursion}{}{}
    \SetKwProg{Asearch}{Asearch}{}{}
    \SetKwProg{Csearch}{ComponentSearch}{}{}
    \SetKwProg{Updist}{Updist}{}{}
    \SetKw{recdijkstrafunc}{recursive\_step}
    \SetKw{Asearchfunc}{Asearch}
    \SetKw{Csearchfunc}{ComponentSearch}
    \SetKw{Updistfunc}{Updist}{}{}
    \SetKwProg{Init}{Main}{}{}
    \SetKwData{arcorder}{arc\_order}
    \SetKwData{indegree}{indegree}
    \SetKwData{count}{count}
    \SetKwData{dist}{dist}
    \SetKwData{parent}{parent}
    \SetKwData{weight}{weight}
    \SetKwData{queue}{queue}
    \SetKwData{component}{component}
    \SetKwData{pred}{root}
    \SetKwData{inarcs}{in-arcs}
    \SetKwInOut{Input}{Input}
    \SetKwInOut{Output}{Output}
    \Input{SSSP algorithm $\mathcal{SP}$ and graph $G = (V,E,s)$}
    \Output{Shortest distances and path tree given by \dist and \parent}
        Create dictionaries $\dist$, $\parent$, $origin$\;
        Initialise $\dist[s] = 0$, and $\dist[v] = \infty$ and $\parent[v] \gets None$ for all $v$\;
        Initialise $v.find = v$ and $v.find.weight = 0$ for each node $v$ in $G$\; \label{alg:rec sssp2: init_find}
        Compute $T_{AC} \gets \texttt{A-C\_tree}(G)$\;
        \mainfunc{$(s)$}\;
        Compute \dist using the shortest path tree \parent\;
        \Return \dist, \parent\;
    \main{$(a)$}{
        $K_1,K_2,\dots,K_m \gets T_{AC}[a]$\;
        \For{each $i$ from 1 to $m$ \label{iterate through Ki}}{
            Create empty graph $G_{K_i}$, with node set $K_i \cup \{a\}$\;
            \For{each $u \in K_i$}{
                \mainfunc($u$)\;
            }
            \For{each $u \in K_i$}{
                \For{each arc $(x,u)$ in $G$ with weight $W$ \label{predecessors}}{
                    $v, d \gets find(x)$\; \label{alg:rec sssp2:call_find}
                    \eIf{$(v,u)$ in $G_{K_i}$}{
                        \If{$d + W < weight_{G_{K_i}}(v,u)$}{
                            Update weight of $(v,u)$ in $G_{K_i}$ to $d + W$\;
                            $origin[v,u] \gets x$\;
                        }
                    }{
                        Add arc $(v,u)$ to $G_{K_i}$ with weight $d + W$\;
                        $origin[v,u] \gets x$\;
                    }
                }
            }
            $dist\_a\_K_i,\ parent\_a\_K_i \gets \mathcal{SP}(G_{K_i},a)$\tcp*{SP:s to $K_i$ from $a$} \label{alg: rec sssp2: SP}
            \For{each $u\in K_i$}{ \label{alg:rec sssp2: last_for_loop}
                $(u.find, u.find.weight) \gets (a,dist\_a\_K_i[u])$ \label{update u}\; \label{alg:rec sssp2: pointer_update}
                $\parent[u] \gets
                origin[parent\_a\_K_i[u],u]$ \label{rec sssp2: recover origin}
            }
        }
    }
    \caption{\texttt{Recursive SSSP}}
    \label{alg:rec sssp2}
\end{algorithm}

\subsection{Recursive SSSP} \label{sec: rec SSSP}

In Algorithm~\ref{alg:recursive_dijkstra_main}, we integrated the A-C tree into the recursive structure of Dijkstra's algorithm. Here, we show that the A-C tree can be used to speed up any SSSP algorithm in a black-box way, using a slightly more complicated method shown in Algorithm~\ref{alg:rec sssp2}. It operates on the same principles as Recursive Dijkstra, in that we will recursively find shortest paths in each subtree of the A-C tree, except that we can no longer assume that when we recursively find the shortest paths from a node $u$, the shortest path from $s$ to $u$ is already known. To account for this, we use (1) an augmented merge-find data structure (called at line \ref{alg:rec sssp2:call_find} using $find$) to hold the roots of recursive shortest-path-subtrees, see Appendix~\ref{sec: merge-find} for details, and (2) build an auxiliary graph $G_{K_i}$ for each component $K_i$ in the A-C tree on which we will then separately run the SSSP algorithm (line \ref{alg: rec sssp2: SP}). Because we run SSSP on a modified graph, we must convert the shortest paths found in $G_{K_i}$ back into paths in $G$. We do this by storing the `true' tail of an arc $\arc{v}{u}$ in $G_{K_i}$ in a dictionary $origin[v,u]$, from which we recover the true shortest-path-parents on line~\ref{rec sssp2: recover origin}. The formal result is given in the following theorem. We defer the proof to Appendix~\ref{sec: proofs sec 6}, though the principles are the same as for Recursive Dijkstra.

\begin{thm}\label{thm: recursive_SSSP}
    Given an SSSP algorithm $\mathcal{SP}$ that runs in time $f_{\mathcal{SP}}(n,m)$ on a graph with $n$ nodes and $m$ arcs, Algorithm~\ref{alg:rec sssp2} solves SSSP in time $O(m \alpha(n) + \sum_{K_i} f_{\mathcal{SP}}(n_i,m_i))$, where the sum is taken over all components $K_i$ of the A-C tree, $n_i$ ($m_i$) is the number of internal nodes (respectively arcs\footnote{Specifically, an arc is `internal' in $K_i$ with respect to $T[a]$ if it corresponds to an arc between two nodes of $K_i$ in $G_a$.}) in component $K_i$, and $\alpha(n)$ is the inverse Ackermann function.
\end{thm}
\begin{proof}
    See Appendix~\ref{sec: proofs sec 6}.
\end{proof}


\begin{corol}\label{corol: sparse_graph_algorithm}
Applying Algorithm~\ref{alg:rec sssp2} to the SSSP algorithm in \cite{duan2025breaking} (with time complexity $O(m\log^{2/3}(n)$), results in an SSSP algorithm with improved time complexity $O(m \alpha(n)+m\log^{2/3}(\nestwidth(G)))$.
\end{corol}
\begin{proof}
    See Appendix~\ref{sec: proofs sec 6}.
\end{proof}

\section{Example} \label{sec: example}

Let us return to Figure~\ref{fig: intro example} from the introduction. Running algorithm \texttt{A-C\_tree} (Algorithm \ref{alg:ac decomposition}) gives us the A-C tree of $G$ shown in Figure~\ref{fig:AC_tree}. By Theorem~\ref{thm: correctness}, this A-C tree defines a minimal-width nesting decomposition of $G$ with width 3. The graph $G$ is strongly connected, but its A-C tree is quite non-trivial with a small nesting width. Intuitively, $G$ has an `almost acyclic' structure.

The algorithm also identifies the nesting bottleneck in this structure, which is the strongly connected set $\{e,d\}$. 
Indeed, when computing shortest paths using Recursive Dijkstra (Algorithm \ref{alg:recursive_dijkstra_main}), the maximum queue length used in the algorithm is equal to $\max_i |K_i| = |\{e,d\}|  =2$, where maximum is taken over all strongly connected components $K_i$ in the A-C tree. This is a significant improvement compared to a na\"ive execution of Dijkstra's algorithm, with a single large queue of length $9$ containing all nodes.


\section{Conclusion} \label{sec: conclusions}

\subsection{Summary}


In this paper, we introduced a graph decomposition called the A-C tree to obtain beyond-worst-case time complexity for any SSSP algorithm. The A-C tree breaks up a single-source graph $G$ into a recursively nested sequence of strongly connected components in topological order. We showed that the A-C tree is optimal in the sense that it achieves the nesting width $\nestwidth(G)$, the minimal width of any nesting decomposition of $G$. We also presented a linear-time algorithm for constructing the A-C tree. Using this linear-time algorithm, we can transform an SSSP algorithm by first computing the A-C tree, and then run the SSSP algorithm on the A-C tree, for reduced time complexity. Indeed, we showcased this transformation with state-of-the-art Dijkstra's algorithm and the recent sparse graph algorithm from \cite{duan2025breaking}, with improved time complexities $O(m+n\log(\nestwidth(G)))$ and $O(m\alpha(n)+m\log^{2/3}(\nestwidth(G)))$, respectively. In particular, SSSP for graphs with bounded nesting width can be solved in linear~time, an important step towards a universally optimal SSSP algorithm.

\subsection{Future Work}

Our results demonstrate that even on this classical and fundamental problem, new graph ideas can still lead to new insights. Among other things, we highlight how, despite the important progress on the complexity of DO, the arguably more natural SSSP problem is open in many cases. The fact that a combination of well-known linear-time graph algorithms can lead to non-trivial advances on this problem suggest that there is room for significant progress in future work. Direct follow-up questions are as follows.

First, classifying the classes of graphs with bounded nesting width is important. In a sense, these graphs are `somewhat acyclic', generalising DAGs, and as a result, SSSP can be solved in linear time. 
As the example in this paper illustrate, 
these classes of graphs are surprisingly rich, and are therefore worth classifying. Second, it would be interesting to understand the frequency that graphs with low nesting width occur in (sparse) random graphs. 
Third, our approach is entirely structural and does not depend on the arc weights. This resembles \emph{universal optimality} in the sense of \cite{haeupler2024universal}. Our results could be a first step towards building a universally optimal algorithm for SSSP, similar to the results of \cite{haeupler2024universal} for DO. Fourth, there are several related problems to SSSP, such as graphs with stochastic transitions between nodes, that could potentially apply our method for improved performance. Fifth, our independence from weights could also be relevant in application domains where, for example, the weights can change over time but the graph topology does not. One motivating example is road networks, where traffic conditions change over time.


\section*{Acknowledgements}

The authors would like to thank Mihalis Yannakakis and Christos Papadimitriou for valuable feedback on the paper.

\bibliographystyle{plain}
\bibliography{references,refs_additional}

\appendix

\section{Proofs from Section~\ref{sec:prelims}} \label{sec: proof sec 2}

\begin{lem}[Lemma~\ref{lem: intersection}]
    If $M$ and $H$ are overlapping modules (they intersect but neither contains the other), then $M\cup H$ and $M\cap H$ are both modules.
\end{lem}
\begin{proof}
    Let $M$ and $H$ be modules with sources $m$ and $h$ respectively. Let $x\in M\cap H$. We know $m\geq_s x$ and $h\geq_s x$ so either (1) $m\geq_s h$ or (2) $h\geq_s m$ by Lemma~\ref{tree order lemma}. Assume w.l.o.g that (1) holds. Then any path $\pth{s}{x}$ contains first $m$, then $h$, then $x$. Because the subpath $\pth{h}{x}$ does not contain $m$, but $x\in M$, we must have $h\in M\cap H$. \emph{Claim}: if $m\in M\cap H$ and $h\in M\cap H$ then $m=h$. To prove the claim, suppose that $h\in M\cap H$ and $m\in M\cap H$. Since there are paths $\pth{s}{m}$ and $\pth{s}{h}$, there must be an arc $\arc{u}{v}$ with $v\in M\cup H$ and $u\not\in M\cup H$.\footnote{Unless $s \in M \cup H$. However, if $s \in M \cup H$, then $m=s$ or $h=s$ and since both $m$ and $h$ are in $M \cap H$, $s \in M \cap H$. Thus, $s$ must be the source of both $M$ and $H$ and the claim follows.} We know that $v = m$ or $v=h$, so $v \in M\cap H$. But then because $M$ and $H$ are modules, we must have $v=m=h$, completing the claim. If follows from this claim that either $m=h$, or $m\in M\setminus H$ and $h\in M \cap H$. 
    Now let $\arc{u}{v}$ be such that $u\not\in M\cap H$ and $v \in M\cap H$. If $u\not\in M$, then $v=m$ as $M$ is a module, and so $m\in M\cap H$ implies $m=h$. If instead we have $u\not\in H$, this implies that $v=h$ as $H$ is a module. Hence $M\cap H$ is a module with source $h$. Similarly, let $\arc{u}{v}$ be such that $u\not\in M\cup H$ and $v \in M\cup H$. If $v\in M$, then $v=m$ as $M$ is a module. If $v\in H$, then $v=h$ as $H$ is a module. However, $h\in M\cap H$, so in this case $v\in M$ which implies that $v = m =h$. Hence $M\cup H$ is a module with source $m$.
\end{proof}

\begin{prop}[Proposition \ref{prop: generalise_dags}]
For any graph,\footnote{For $G$ with more than one node. When $G$ has just one node, then $\nestwidth(G)=1$.} $2 \leq \nestwidth(G) \leq n$. In particular, for a DAG $G$, the lower bound $\nestwidth(G)=2$ is attained. In this sense, $\nestwidth(G)$ can be seen as a measure of how acyclic $G$ is.
\end{prop}

\begin{proof}
    As shown in Section \ref{sec: nesting_width}, by looking at the equivalent tree form of a nesting decomposition, we have $\nestwidth(G) \leq n$. Similarly, the equivalent tree form of a nesting decomposition of $G$ must have a node with at least two children (given that $G$ has more than one node). Thus, $\nestwidth(G) \geq 2$. To show that $\nestwidth(G) = 2$ for a DAG $G$, apply Theorem \ref{thm: correctness} and note that there cannot be any strongly connected components other than singletons. Thus, by the proof of Theorem \ref{thm: correctness}, $\nestwidth(G) = 1 + \max_{a,i} |S^a_i| = 2$. This concludes the proof.
\end{proof}

\section{Proofs from Section~\ref{sec: rec Dijkstra}} \label{sec: proofs sec 6}

\begin{thm}[Theorem~\ref{thm: recursive_SSSP}]
    Given an SSSP algorithm $\mathcal{SP}$ that runs in time $f_{\mathcal{SP}}(n,m)$ on a graph with $n$ nodes and $m$ arcs, Algorithm~\ref{alg:rec sssp2} solves SSSP in time $O(m \alpha(n) + \sum_{K_i} f_{\mathcal{SP}}(n_i,m_i))$, where the sum is taken over all components $K_i$ of the A-C tree, $n_i$ ($m_i$) is the number of internal nodes (respectively arcs) in component $K_i$, and $\alpha(n)$ is the inverse Ackermann~function.
\end{thm}
\begin{proof}
    \textbf{Time complexity:}
    The algorithm begins with some initialisation, all of which can be performed in linear time because the A-C tree can be computed in linear time. We then call $TreeRecursion$ on the source $s$.

    Now consider a call to $TreeRecursion(a)$. We iterate through each component $K_i$ on line \ref{iterate through Ki}. Consider the iteration for $K_i$. After creating the graph $G_{K_i}$, we iterate through each $u\in K_i$ and do the following: (1) recurse on $u$; (2) look at each predecessor of $u$, call $find$, and update an arc in $G_{K_i}$; (3) call $\mathcal{SP}$ on $G_{K_i}$. Finally we again iterate through the output of $\mathcal{SP}$ for each $u\in K_i$.
    
    Ignoring for a moment the recursive call $TreeRecursion(u)$, observe that all other operations inside this loop work only on nodes in $K_i$. Similarly, the recursive call $TreeRecursion(u)$ operates only on nodes which are below $u$ in the A-C tree; none of these are in $K_i$. Because each node in the graph is a member of exactly one component $K_i$, the time complexity is simply the sum of the time complexities of these operations over each component $K_i$. 

    Let $n_i$ be the number of nodes in $K_i$, $m_i$ the number of arcs whose head is in $K_i$ (the tail may or may not be in $K_i$). Iterating through $K_i$ takes $O(n_i)$ operations. Iterating through the predecessor arcs and updating arcs in $G_{K_i}$ takes $O(m_i\alpha(n))$ operations, because we call $find$ once for each arc with time complexity given by Lemma~\ref{lem:find:complexity}. Finally, running $\mathcal{SP}$ takes $O(f_{\mathcal{SP}}(n_i,m_i))$ steps. Summing this over every $K_i$ in the A-C tree gives us $O(\sum_{K_i} (n_i + m_i \alpha(n) + f_{\mathcal{SP}}(n_i,m_i))) = O(n + m\alpha(n) + \sum_{K_i}f_{\mathcal{SP}}(n_i,m_i))$, as required (note that $n\in O(m)$ because the graph is connected).

    \textbf{Correctness:} For brevity, abbreviate `shortest path' as SP.
    In order to produce the true SP tree, we need that upon completion of $TreeRecursion(s)$, $parent[u]$ is equal to the true SP-parent of $u$, for each node $u$. We shall prove this in a series of claims, using~induction. 

    \emph{Inductive hypothesis 1 (IH1):} Assume for induction that after $TreeRecursion(u)$ is completed, for all $w<_s u$ we have that $parent[w]$ is equal to the SP-parent of $w$. The base case holds trivially, because if $u$ has no children in the A-C tree then there is no $w<_s u$.

    Suppose we are running $TreeRecursion(a)$, and we assume that the inductive hypothesis (IH1) holds for all children of $a$ in the A-C tree. We want to show that (IH1) also holds for $a$ once the execution finishes. We shall do this using a second induction step over the components $K_i$ that are children of $a$ in the A-C tree.

    \emph{Inductive hypothesis 2 (IH2):} At the start of the $i$th iteration of line~\ref{iterate through Ki} in $TreeRecursion(a)$, we have $parent[w]$ equal to the SP-parent of $w$ for all $w\leq_s v$, for all $v\in K_j$ with $j<i$.

    In other words, we assume that the first inductive hypothesis holds for nodes which are children of the components $K_j$ which we have already handled in this loop. If we show that (IH2) holds for all components $K_i$, it will complete the inductive proof of (IH1). Again the base case is trivial, so we shall now assume (IH2) holds up to the $i$th iteration.

    To complete our setup, we must also discuss the operation of the $find$ function. Observe that at any point in the execution of Recursive SSSP (Algorithm~\ref{alg:rec sssp2}), $parent$ defines a directed forest on the nodes in $G$. Similarly, $find(x)$ (shown in Algorithm~\ref{alg:find}) works by finding the root of the tree defined by the pointers $x.find$, and $d$ is computed by summing the weights associated with each such pointer. We need to show that, at any point in the execution of the algorithm, the root of the $find$-tree is always equal to the root of the $parent$-tree, and likewise that the sum of weights on the paths in both trees from their roots are equal. 

    \emph{Inductive hypothesis 3 (IH3):} For any $x$, whenever $find(x)$ is called it returns $v,d$ such that (1) $v$ is the root of the $parent$ tree containing $x$ and (2) $d$ is the true distance from $v$ to $x$.

    For the base case, observe that initially, $parent[v]$ is None for all $v$, so $parent$ is a forest consisting of disjoint nodes. We also initialise $v.find = v$, $v.find.weight = 0$, so $find(v) = v,0$, as required. 
    
    Now we suppose that, up to some iteration $i$ within a call to $TreeRecursion(a)$, the hypothesis (IH1), (IH2) and (IH3) hold. We shall show that all of them continue to hold after completion of this iteration.

    On iteration $i$ we will construct a graph $G_{K_i}$, run $\mathcal{SP}$ on this graph and use the shortest paths in $G_{K_i}$ to extract the true shortest paths to nodes in $K_i$. We shall use the following definition: given two nodes $v$ and $u$, we call a path from $v$ to $u$ a $\downarrow v$-\emph{path} if for all nodes $w\neq u$ on this path, we have $w\leq_s v$. The idea is that a shortest path in $G$ factors into a collection of $\downarrow v$-paths, which correspond to arcs in $G_{K_i}$. Our main fact is that $G_{K_i}$ has the following properties:
    
    \emph{Claim:} (Construction of $G_{K_i}$) 
    For $v,u\in K_i$, there is an arc $(v,u)$ in $G_{K_i}$ if and only if there is a $\downarrow v$-path in $G$ from $v$ to $u$. The weight on $(v,u)$ in $G_{K_i}$ is the minimum length of all $\downarrow v$-paths to $u$. Finally, $origin[v,u] = z$ where $z$ is the predecessor of $u$ on the shortest $\downarrow v$-path from $v$ to $u$. For arcs $(a,u)$, the same properties hold corresponding to $\downarrow a$-paths in $G$ where $u$ is the only node in $K_i$.
    
        \emph{Proof of claim:} We first construct an empty graph $G_{K_i}$. We then run $TreeRecursion(u)$ on each $u\in K_i$, so by (IH1) we know that $parent[w]$ is the SP-parent of $w$ for each $w<_s u$ for all $u\in K_i$. Next, on line~\ref{predecessors}, we examine each predecessor $x$ of each $u\in K_i$. Because all paths into $K_i$ contain $a$, and the $K_i$ are topologically ordered from $a$, we have that either (1) $x\leq_s w$ for $w\in K_i$; (2) $x\leq_s w$ for $w\in K_j$, $j< i$; or (3) $x = a$.

        By (IH3), $find(x) = v,d$, where $v$ is the root of the $parent$ tree from $x$. If (2) or (3), the root of the $parent$ tree from $x$ is $a$, by (IH1) and (IH2). If (1), the root of the $parent$ tree from $x$ is $w\in K_i$, by (IH1). In both cases, observe that $x\leq_s v$, and by (IH3) $d$ is the shortest distance from $v$ to $x$. We then update the weight of the arc $(v,u)$ in $G_{K_i}$ if $d+W$ is less than its current weight. The value $d + W$ is exactly the length of the shortest path from $v$ to $x$ followed by the arc $(x,u)$. This is a $\downarrow v$-path to $u$, because $x\leq_s v$. Every $\downarrow v$-path from $v$ to $u$ must have this form, and because we check all predecessors of $u$ we must add an arc $(v,u)$ to $G_{K_i}$ if such a path exists. Further, because subpaths of shortest paths are shortest paths, the shortest $\downarrow v$-path from $v$ to $u$ must consist of a shortest path from $v$ to $x\leq_s v$ followed by an arc $(x,u)$. Hence after checking all such paths we find the minimum distance, and this becomes the weight on $(v,u)$ in $G_{K_i}$. Finally, observe that we always set $origin[v,u]$ to be the penultimate node on the shortest such path, completing the claim.
    
    To complete the proof, consider the structure of a shortest path in $G$ from $a$ to a node $u\in K_i$. We will break it into subpaths based on the presence of nodes in $K_i \cup \{a\}$. Let $S_0$ denote the prefix, beginning at $a$ (which we denote $w_0$) and containing all nodes up to (but not including) the first node $w_1$ in $K_i$ (which could be $u$). After that, the next subpath $S_1$ will consist of $w_1$, followed by all nodes up to but not including the next node $w_2$ in $K_i$. We define the next segment $S_2$ analogously, and continue in this way until our path has the form: $S_0,S_1,S_2,\dots,u$. The first sequence $S_0$ is nonempty because it contains $a$, but the remaining sequences $S_i$ could be empty.

    Each $S_j$ begins at $w_j$ and continues until the next node $w_{j+1}\in K_i$. Observe that all nodes in $S_j$ must be in $\downarrow w_j$, because they have a path to $w_{j+1}\in K_i$ that does not contain $a$. Also because subpaths of shortest paths are also shortest paths, each of these subpaths is a shortest $\downarrow w_j$-path from $w_j$ to $w_{j+1}$. By the claim above, there must be an arc $(w_j,w_{j+1})$ in $G_{K_i}$, and the weight of this arc is equal to the total length of the path $S_j$ (the sum of the weights on the arcs out of these nodes). Further, if the final node in $S_j$ is $x$, then $origin[w_j,w_{j+1}] = x$.

    Hence there is an associated path $a,w_1,w_2,\dots,u$ in $G_{K_i}$, which has the same total length. Because all paths in $G_{K_i}$ correspond to paths in $G$, this path is also a shortest path in $G_{K_i}$. Hence for all $u\in K_i$, $dist\_a\_K_i[u]$ is equal to the true distance from $a$ to $u$. Consequently, $parent\_a\_K_i[u]$ is equal to either $a$ (if the path is exactly $a,u$) or otherwise the final $w_k$. In both cases, by the claim, $origin[w_k,u] = x$, the predecessor of $u$ on the final segment of the path, which is its true SP-parent in $G$. This is true for all $u\in K_i$, completing the proof of (IH2) and hence (IH1).

    For each $u\in K_i$, we assign it a $parent[u]$ that is not itself. Hence the root of the $parent$-tree for each such $u$ must lie outside $K_i$, and so by (IH2) it must be $a$ (note that during execution of $TreeRecursion(a)$, $parent[a] = None$). At line~\ref{update u} we set $u.find$ to be $a$, so on the subsequent iteration the root of the $find$-tree remains equal to the root of the $parent$ tree for all nodes. Finally, we set $u.find.weight$ to $dist\_a\_K_i[u]$, which we established above is equal to the true distance from $a$ to $u$. Hence (IH3) also holds on the subsequent iteration, completing its proof by induction.
\end{proof}

\begin{proof}[Proof of Corollary~\ref{corol: sparse_graph_algorithm}]
    We upper-bound $O(m \alpha(n) + \sum_{K_i} f_{\mathcal{SP}}(n_i,m_i))$ as follows
    \begin{align*}
    m \alpha(n) + \sum_{K_i} f_{\mathcal{SP}}(n_i,m_i) \leq \\
    m \alpha(n) + \sum_{K_i} f_{\mathcal{SP}}(\nestwidth(G),m_i)  \in \\
    O(m \alpha(n) + \sum_{K_i} m_i\log^{2/3}(\nestwidth(G)))) \leq \\ O(m\alpha(n)+m\log^{2/3}(\nestwidth(G)))).
    \end{align*}
    The result now follows from Theorem \ref{thm: recursive_SSSP}. Note that the resulting algorithm is indeed faster since $\alpha(n)$ increases slower than $\log^{2/3}(n)$, and $\nestwidth(G) \leq n$.
\end{proof}

\section{Merge-find structure} \label{sec: merge-find}

In this section, we provide the details of the merge-find structure, used by Algorithm~\ref{alg:rec sssp2} to quickly get shortest distances in the A-C tree. 

The idea is very similar to the standard merge-find structure \cite{tarjan1984worst} using disjoint forests with path compression, except that each pointer in the tree now additionally holds a weight value.\footnote{For the merge-find structure, we use the standard terminology `pointer' instead of `arc'. This also avoids potential confusion with the arcs in the A-C tree.} Concretely, each node $u$ in $G$ has a pointer to some ancestral node $v$ higher up in the A-C tree, with weight equal to the distance from $v$ to $u$ (note the order). By updating the pointer of $u$ whenever needed, Algorithm~\ref{alg:rec sssp2} can quickly retrieve relevant shortest distances in the A-C tree. Crucially, as we will see, the additional weight feature does not increase the amortised time complexity of the merge-find structure.



We continue by describing the details of the merge-find structure and its main function $find$ in Section \ref{find:description}. Then, we prove the time complexity in Section \ref{find:complexity}.

\begin{algorithm}
    \SetAlgoLined
\SetKwInOut{Input}{Input}
\SetKwInOut{Output}{Output}
\SetKwData{found}{found}
\SetKwData{parent}{parent}
\SetKwData{current}{current}
\SetKwData{graph}{graph}
\SetKwProg{compress}{\texttt{RecursiveDominanceGraphs}}{}{}
\SetKw{comp}{\texttt{RecursiveDominanceGraphs}}
\Input{Node $u$ in a graph $G$}
\Output{Updated pointer $u.find$ with weight $u.find.weight$ for $u$.}
Create empty list $path$ and append $u$\;
$x \gets u$\;
\While{$x.find \neq x$}{ \label{alg:find:while}
    $x \gets x.find$\;
    Append $x$ to $path$\;
}
$x_{next} = x.find$\;
\For{$x$ in $path$ in reversed order}{
    $x.find \gets x_{next}.find$\;
    $x.find.weight \gets x.find.weight + x_{next}.find.weight$ \tcp*{Extra feature} \label{alg:find:update_weight}
    $x_{next} = x$\;
}
\Return $(u.find, u.find.weight)$\; \label{alg:find:return}
\caption{\texttt{find}}
\label{alg:find}
\end{algorithm}

\subsection{Description}\label{find:description}
When running Algorithm~\ref{alg:rec sssp2}, the merge-find structure keeps track of relevant shortest distances using the function $find$, summarised by Algorithm~\ref{alg:find}. More precisely, given a node $u$ of $G$, $find(u)$ returns the shortest distance to $u$ from an ancestor $v$ in the A-C tree, where $v$ is the \emph{current highest ancestor} in the A-C tree with known shortest distance to $u$. We explain this procedure in the following steps.

\textbf{Initialisation.} Initially, we only know the shortest distance from $u$ to itself. Thus, the merge-find structure is initialised with each node $u$ of $G$ having a pointer $u.find = u$ to itself with zero weight, i.e., $u.find.weight = 0$ (line \ref{alg:rec sssp2: init_find} in Algorithm~\ref{alg:rec sssp2}). The weight is zero since the shortest distance from $u$ to itself is zero.

\textbf{First update.} The pointer $u.find = u$ remains unchanged until the shortest distance from $a$ to $u$ becomes known, where $a$ is the parent of $u$ in the A-C tree. This happens at line \ref{alg: rec sssp2: SP} in Algorithm~\ref{alg:rec sssp2}, with the shortest distance from $a$ to $u$ given by $dist\_a\_K_i[u]$. Correspondingly, $u.find$ gets updated to
$(u.find, u.find.weight) \gets (a,dist\_a\_K_i[u])$ in line \ref{alg:rec sssp2: pointer_update}. The current highest ancestor in the A-C tree with known shortest distance to $u$ is now $a$. 

\begin{rmk}[Merge] \label{rmk:merge}
    In merge-find terminology, the combination of the for-loop in line \ref{alg:rec sssp2: last_for_loop} with line \ref{alg:rec sssp2: pointer_update} in Algorithm \ref{alg:rec sssp2} is just a merge-operation, where we `merge' all nodes in $K_i$ by redirecting their pointers to a common node (in this case node~$a$). 
\end{rmk}

\textbf{Later updates.} For later updates of the pointer $u.find$, the merge-find structure uses $find(u)$, given by Algorithm \ref{alg:find}. More precisely, $find(u)$ always returns (i) the current highest ancestor in the A-C tree $v$ with known shortest distance to $u$, and (ii) the distance $d$ from $v$ to $u$. This is done in two stages. In the first stage, $find(u)$ starts from $u$ and iteratively follows pointers upwards in the A-C tree until reaching a node $v$ that points to itself (the while-loop in Algorithm \ref{alg:find}). Let this path of encountered nodes be
\begin{equation}\label{eq:pointer_path}
u \rightarrow_p x_1 \rightarrow_p x_2 \rightarrow_p \dots \rightarrow_p x_{k} \rightarrow_p v 
\end{equation}
where $u \rightarrow_p x_1$ means $u.find = x_1$. Since $v$ points to itself, the distance to $v$ from the parent $a$ of $v$ is currently unknown. Hence, the node $v$ satisfies (i). In the second stage, $find(u)$ updates the pointers for \emph{all} encountered nodes on the path in Equation \eqref{eq:pointer_path}; this makes sense since $v$ satisfies condition (i) not only for $u$ but for all nodes on this path. The update is done by iteratively going over the path in reverse order using ($x_0 = u$ and $x_{k+1} = v$)
\begin{equation}\label{eq:find_iteration}
(x_i.find, \; x_i.find.weight) \gets (x_{i+1}.find, \; x_i.find.weight+x_{i+1}.find.weight)
\end{equation}
for $i=k,k-1,\dots,0$ (the for-loop in Algorithm \ref{alg:find}). Thus, at the end of $find(u)$, each node $x_i$ points at $v$ with weight $x_i.find.weight$ equal to the shortest distance to $x_i$ from $v$, where the latter is a consequence of the construction of the A-C tree.\footnote{To see this, note that $x_i$ is dominated by $x_{i+1}$, so the shortest path from $v$ to $x_i$ must visit the nodes (in order) $v, x_k, \dots,x_{i+1},x_i$. Assume by induction that $x_i.find.weight$ is the shortest distance from $x_{i+1}$ to $x_i$ (this is true in the base case when $x_i.find$ is first updated). Then, the distance from $v$ to $x_i$ is $\sum_{l=i}^{k+1} x_i.find.weight$, due to domination. This is exactly what we get by iteratively applying Equation~\eqref{eq:find_iteration}.} In other words, all pointers along the encountered path are now up-to-date. 

The reason we update all nodes is to improve speed. For example, any node $x_i$ on the path in Equation \eqref{eq:find_iteration} is such that $find(x_i)$ runs in $O(1)$ steps just after $find(u)$ has been called, since the pointer of $x_i$ is already up-to-date pointing to $v$. More generally, after running $find(u)$, any node $\tilde{u}$ with ancestor $x_i$ will skip any of the nodes between $x_i$ and $v$ when calling $find(\tilde{u})$. This is because the pointer of $x_i$ now goes directly to $v$, illustrated by Figure~\ref{fig:find_combined}. The result is a near-constant time complexity $O(\alpha(n))$ per function call, shown next.

\begin{figure}[h]
\centering
\includegraphics[width=0.99\linewidth]{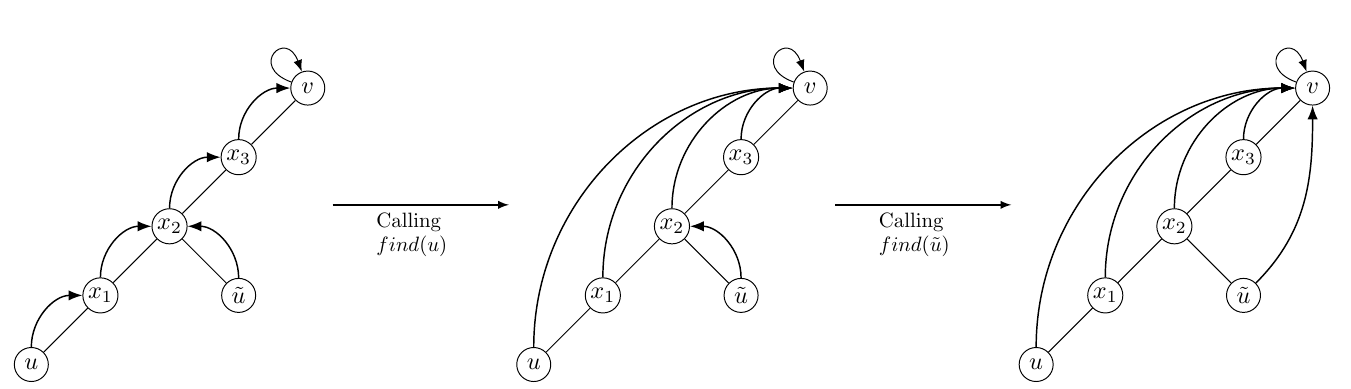}
\caption{Procedure illustration of $find$. Undirected edges corresponds to the tree hierarchy of the A-C tree (not depicting connected components) while each arrow represent a pointer from a node (not depicting weights). Left: $find$ has not been called. Middle: $find(u)$ has been called. Right: $find(\tilde{u})$ has been called. Note that $find(\tilde{u})$ can skip node $x_3$ and go directly from $x_2$ to $v$, due to the previous call of $find(u)$. Such skips result in the near-linear time complexity $O(m \alpha(n))$ of the merge-find structure.}
\label{fig:find_combined}
\end{figure}  

\subsection{Time complexity}\label{find:complexity}
\begin{lem}\label{lem:find:complexity}
The amortised time per function call of $find(x)$ in Algorithm~\ref{alg:rec sssp2} is $O(\alpha(n))$.
\end{lem}
\begin{proof}
    The only additional structure compared to a standard merge-find structure is the weight feature. Concretely, $find(x)$ (Algorithm \ref{alg:find}) is identical to one for a standard merge-find structure except lines \ref{alg:find:update_weight} and \ref{alg:find:return} that account for the weights. These modifications do not add any time overhead of $find(x)$ compared to standard merge-find, i.e., we still have $O(\alpha(n))$ in amortised time per function call \cite{tarjan1984worst}.
\end{proof}

\section{Bounded nesting width and sparse graphs are distinct} \label{sec: nesting width vs sparsity}
In this section, we prove that bounded nesting width graphs and sparse graphs form \emph{distinct} classes, formalised by the following proposition. Thus, any universally optimal SSSP algorithm needs to be optimal for \emph{both} sparse graphs and bounded nesting width graphs. This indicates the importance of bounded nesting width graphs as a separate crucial graph class for SSSP~algorithms. 

\begin{prop}\label{prop: nesting_width_and_sparsity}
Consider the bounded nesting width graph class $C_{BNW} = \{G=(V,E,s): \nestwidth(G) \leq M\}$ and the sparse graph class $C_{S} = \{G=(V,E,s): m \leq L\cdot n\}$ with constants $M\geq 2$ and $L>1$. Then $C_{S} \nsubseteq C_{BNW}$ and $C_{BNW} \nsubseteq C_{S}$.
\end{prop}

\begin{proof}
    We first show that $C_{BNW} \nsubseteq C_{S}$. Consider the DAG $G = (V,E,s)$ consisting of nodes $V= \{v_1,v_2,\dots,v_n\}$ indexed in topological order, with arcs $v_i \rightarrow v_j$ for all $j>i$. We then have
\begin{equation*}         
m = \sum_{i=0}^{n-1} i = \frac{n (n-1)}{2} \in \Omega (n^2) 
\end{equation*}
    In particular, given any $L$, we have that $m> L \cdot n$ for sufficiently high $n$. Thus, $G \notin C_{S}$ for sufficiently high $n$. On the other hand, $\nestwidth(G) = 2$ since $G$ is a DAG (Proposition \ref{prop: generalise_dags}), so $G \in C_{BNW}$ (for any $n$). We conclude that $C_{BNW} \nsubseteq C_{S}$.

    Next, we show that $C_{S} \nsubseteq C_{BNW}$. To this end, construct a graph $G =(V,E,s)$ as follows. Let $k>0$ be an integer (to be determined) and consider the complete graph $C_k$ of $k$ nodes where every node has an arc to any other node. This is a subgraph of $G$. For the remaining $n-k$ nodes $\{v_1,v_2,\dots,v_{n-k}\}$ of $G$, arrange them in a (directed) line such that $v_i \rightarrow  v_{i+1}$ for $i=1,\dots,n-k-1$, and let $v_{n-k}$ have arcs to all nodes in $C_k$. It is easy to see that the subgraph $C_k$ cannot be further decomposed by the A-C tree and so, by the proof of Theorem \ref{thm: correctness}, the nesting width is $\nestwidth(G) = 1 + \max_{a,i} |S^a_i| = 1 + |C_k| = 1+k$. Furthermore, the number of arcs in $G$ are
    
    \begin{equation*}
        m = (n-k-1)+k+k^2 = n+k^2-1
    \end{equation*}
    Thus, we want to pick $k=k(n)$ such that (i) $\nestwidth(G) = 1+k>M$ and (ii) $m = n+k^2-1 \leq L \cdot n$ for sufficiently high $n$, since then $G \notin C_{BNW}$ and $G \in C_{S}$. To this end, consider
    \begin{equation}
    k =  \mathrm{int} \Big \{ \sqrt{(L-1)\cdot n-1}  \Big \},
    \end{equation}
    where $\mathrm{int}(x)$ is the integer part of $x \in \mathbb{R}$. By construction, $m = n+k^2-1 \leq L \cdot n$ (for sufficiently high $n$), so fulfills (ii), and since $k$ increases with $n$, we also fulfill (i). Thus, $G \notin C_{BNW}$ and $G \in C_{S}$, and so~$C_{S} \nsubseteq C_{BNW}$.
\end{proof}

\begin{rmk}
The proof of Proposition \ref{prop: nesting_width_and_sparsity} actually shows a somewhat stronger result. Indeed, we can consider the class of \emph{sub-squared} nesting width class given by $C_{NW} = \{G=(V,E,s): \nestwidth(G) \leq M \cdot f(n)\}$ where $f(n) \in o(\sqrt{n})$ is any function of sub-squared growth. Following the proof (changing condition (i) to $\nestwidth(G) = 1+k>M\cdot f(n)$), we still get $C_{S} \nsubseteq C_{NW}$ and $C_{NW} \nsubseteq C_{S}$. Thus, even for this larger class $C_{NW}$ of sub-squared nesting width, nesting width and sparsity still form distinct classes.
\end{rmk}

\section{Superlinear time complexity in Duan et al. \cite{duan2025breaking}} \label{sec: superlinear}
In this section, we prove that the shortest path algorithm from \cite{duan2025breaking} (Algorithm 3, called BMSSP) has superlinear time complexity. We will make use of much of their notation and refer the reader to \cite{duan2025breaking} for details.

\textbf{Graph:} Consider $G=(V,E,s)$, where $V = \{u_0,\dots,u_N\}\cup\{v_0,\dots,v_N\}$ and $s=u_0$ is the source. Here, $s$ and all the $u$:s form a line of arcs
\begin{equation*}
    u_0 \rightarrow u_1 \rightarrow u_2 \rightarrow u_3 \rightarrow \dots \rightarrow u_{N}.
\end{equation*}
We now also have an arc $u_i \rightarrow v_i$ from each $u_i$ to the corresponding $v_i$. (There are no other arcs.) Moreover, the weights are all uniform, except $weight(u_i,v_i) =T \gg N$ with some large number $T>0$. This ensures that the shortest path to $v_i$ is always much greater than the shortest path to any $u_j$. As we will see, this is important for guaranteeing superlinear time~complexity.

Finally, let $l_{top} = \lceil \log(n)/t \rceil$ be the highest level of the BMSSP algorithm, to distinguish it from the level $l$ at which the BMSSP algorithm is currently at. This notation is not found in \cite{duan2025breaking} but will be helpful for our purposes.

\textbf{Key idea.} By the structure of $G$, almost every level $l$ of the BMSSP algorithm will look at least half of the $v$-nodes (i.e., $N/2$ nodes) as it processes the $u$-nodes. Thus, informally, we should have a lower bound on the time complexity given by $\Omega(\frac{N}{2} l_{top}) = \Omega(n \log^{1/3}(n))$, that is, superlinear time~complexity. 

\textbf{Formal proof.} To provide a formal proof, we need the following notions. 

Let $\tau_l := k2^{lt}$ for level $l\geq0$. Intuitively, $\tau_l$ captures how many $u$-nodes we can process at level $l$ in the BMSSP algorithm. Moreover, let $l^*$ be the level such that $\tau_{l^*}=k2^{l^*t}\geq N$ and $\tau_{l^*-1} = k2^{( l^*-1)t}<N$. Note that such a level $l^*$ must exist since $\tau_{l_{top}} = \Theta(kn) = \Theta(kN)$. Informally, $l^*$ will serve as the level when we run out of $u$-nodes to process in the BMSSP~algorithm.

Consider $s_{last}\geq 1$ such that $(s_{last}-1)\tau_{l^*-1}+\tau_{l^*-1}< N$ and $s_{last}\tau_{l^*-1}+\tau_{l^*-1}\geq N$. Due to the definition of $l^*$, such an $s_{last}$ must exist. The parameter  $s_{last}$ will serve as the number of times we call the BMSSP algorithm at level $l^*-1$ before we run out of~$u$-nodes. 

Finally, let $\mathcal{T}$ be the recursion tree formed by all the calls of $BMSSP(l,B,S)$ with $BMSSP(l=l_{top},B=\infty,S=\{u_0\})$ being the root call of $\mathcal{T}$. Let $n_{l,s}$ be the node in $\mathcal{T}$ corresponding to run number $s$ with $BMSSP(l,B,S)$ (for any $B$ and $S$). For example, node $n_{l^*,s_{last}}$ will be the last function call of $BMSSP(l,B,S)$ at level $l=l^*-1$ before we run out of $u$-nodes. Let $\mathcal{T}^* \subseteq \mathcal{T}$ be the tree as follows. The tree $\mathcal{T}^*$ has root node $n_{l^*,1}$ with children $n_{l^*-1,i}$ for $i=1,2,\dots,s_{last}$. Moreover, each $n_{l^*-1,i}$ in $\mathcal{T}^*$ is in turn the root of a subtree identical to the one in $\mathcal{T}$. Thus, $\mathcal{T}^*$ is indeed a subgraph of $\mathcal{T}$. The intuition is that $\mathcal{T}^*$ will be the part of $\mathcal{T}$ where we have not yet processed all $u$-nodes. The tree $\mathcal{T}^*$ will therefore be our object of interest.

With these notions, we have the following useful lemma, showing how the BMSSP algorithm operates on $\mathcal{T}^*$. This is the crucial step to show superlinear time complexity. For ease of reading, we adopt the shorthand notation $s'=s-1$ for any integer $s\geq 1$.

\begin{lem}\label{lemma:superlinear2}
Let $n_{l,s}  \in \mathcal{T}^* \backslash \{ n_{l^*,1}\}$ with $l \geq 1$. Then:
\begin{enumerate}[label=(\roman*)]
\item \textbf{Input:} The input to $n_{l,s}$ satisfies $B \geq T$ and  $S = \{u_{s' \tau_l} \} \cup S_v$, where $S_v$ is a subset of $v$-nodes. Moreover, $s'\tau_l+\tau_l< N$ (intuitively, there are enough $u$-nodes left when going forwards). 
\item \textbf{Calls:} $n_{l,s}$ calls children $n_{l-1,f}$ in the order $f' =s'2^t,s'2^t+1,\dots,s'2^t+2^t-1$ (and only them).
\item \textbf{Successors:} $n_{l,s}$ looks at each of the nodes in $\{v_{s'\tau_l+i}\}_{i=0}^{\tau_l-1}$.
\item \textbf{Output:} $n_{l,s}$ returns $B' = d(u_{s'\tau_l+\tau_l})$ and $U =\{u_{s'\tau_l+i}\}_{i=0}^{\tau_l-1}$, where $U=\cup_{i=1}^{2^t} U_i$ forms a partition of $U$.
\end{enumerate}
\end{lem}

\begin{proof}
    We prove Lemma \ref{lemma:superlinear2} by induction over the tree $\mathcal{T}^*$ considering nodes $n_{l,s}  \in \mathcal{T}^* \backslash \{ n_{l^*,1}\}$ such that $l \geq 1$. Note that $\mathcal{T}^*$ is executed in DFS order by the BMSSP algorithm (starting from $n_{l^*,1}$). The induction will be based on the same DFS order. In fact, there are two relevant orders based on this DFS order: (1) the order that each node in $\mathcal{T}^*$ starts (i.e., the first time you reach the node in the DFS order); and (2) the order that each node in $\mathcal{T}^*$ finishes (i.e., the last time the node is considered in the DFS order). For brevity, we call the two orders the \emph{starting DFS order} and the \emph{finishing DFS order}, respectively. The induction assumption, given by three separate assumptions (IA1-IA3), uses these two orders, where $\tilde{\mathcal{T}}^* := \{n_{l,s}  \in \mathcal{T}^* \backslash \{ n_{l^*,1}\}: l \geq 1\}$:

    \textbf{(IA1):} Consider a node $n_{l,s} \in \tilde{\mathcal{T}}^*$. Then, for any previous node $n_{\tilde{l},\tilde{s}} \in \tilde{\mathcal{T}}^*$ in the starting DFS order, we have that $n_{\tilde{l},\tilde{s}}$ has been called, $n_{\tilde{l},\tilde{s}}$ satisfies (i).

    \textbf{(IA2):} Consider a node $n_{l,s} \in \tilde{\mathcal{T}}^*$. Then, all children of $n_{l,s}$ that have started agree with the order in (ii).

    \textbf{(IA3):} Consider a node $n_{l,s} \in \tilde{\mathcal{T}}^*$. Then, for any previous node $n_{\tilde{l},\tilde{s}} \in \tilde{\mathcal{T}}^*$ in the finishing DFS order, $n_{\tilde{l},\tilde{s}}$ satisfies (iii) and (iv).
    


    We will need the following observation:

    \textbf{Observation 1:} If (i) is satisfied for a node $n_{l,s} \in \tilde{\mathcal{T}}^*$, then the output from the subroutine FindPivots (line 4 in the BMSSP algorithm) satisfies
    \begin{align*}
        P &= \{u_{s' \tau_l} \} \cup S_v \\
        W &\subseteq \{u_{s' \tau_l+i} \}_{i=0}^k \cup \{v_{s' \tau_l+i} \}_{i=0}^{k-1} \cup S_v
    \end{align*}
    \begin{proof}
        Observation 1 follows from how the subroutine FindPivots is executed.
    \end{proof}

    We now conduct the induction proof, starting with the base case:

    \textbf{Base case.} Consider the series of recursion calls from $n_{l_{top},1}$ all the way down to $n_{1,1}$ in $\mathcal{T}$ (with $n_{l^*,1}$ being one of these nodes). Initially, for $n_{l_{top},1}$, we start with $B=\infty$ and $S = \{u_0\}$. Then we get $P=\{u_0\}$ due to an early termination of FindPivots, and therefore $B_i=B=\infty$ and $S_i = S = \{u_0\}$ are the first values we pull from $D$ in the while-loop of $n_{l_{top},1}$. These values are fed to $n_{l_{top}-1,1}$ as input. For $n_{l_{top}-1,1}$, this process repeats, with input $B_i=B=\infty$ and $S_i = S = \{u_0\}$ as input to $n_{l_{top}-2,1}$, and so on for all nodes in the series of calls down to $n_{1,1}$. Moreover, for any node $n_{l,1} \in \tilde{\mathcal{T}}^*$ we have $\tau_l <N$. To see this, note that $\tau_{l^*-1} <N$ due to the definition of $l^*$ and since $\tau_l < \tau_{l^*-1}$ for any $l \leq l^*-1$. We conclude that each node $n_{l,1} \in \tilde{\mathcal{T}}^*$ satisfies (i). This serves as the base case for (IA1).

    Next, we show the base case for (IA2) and (IA3). To this end, consider $n_{1,1}$. It will call $n_{0,1}$ with input $S_i = \{u_0\}$ and $B_i = \infty$. The node $n_{0,1}$ performs a run with Dijkstra's algorithm (using binary heap), starting from $u_0$ and reaches the nodes $U_0 = \{u_i\}_{i=0}^{k-1}$ and $B' = d(u_k)$, which are returned to $n_{1,1}$. Returning to $n_{1,1}$, we consider the successors of $U_0$, that is, we look at $\{v_i\}_{i=0}^{k-1}$ and $\{u_k\}$. We also add $\{u_k\}$ and $\{v_i\}_{i=0}^{k-1}$ to $D$ (since $B_i' = d(u_k)$ and $B = \infty$). (*) Next, we pull $S_i = \{u_{k}\}$ from $D$ (since it is the smallest element in $D$) with bound $B_i \geq T$ (since the last remaining element in $D$ is a $v$-node) and input this to $n_{0,2}$. Similarly to $n_{0,1}$, $n_{0,2}$ returns $U_0 = \{u_i\}_{i=k}^{2k-1}$ and $B' = d(u_{2k})$. Returning to $n_{1,1}$, we look again at the successors $\{v_i\}_{i=k}^{2k-1}$ and $\{u_{2k}\}$ and add $\{u_{2k}\}$ to $D$ (this is now the smallest element in $D$). The step from (*) then repeats. Indeed, by an easy induction proof, the $f$:th call of $n_{1,f}$ will return $U_0 = \{u_i\}_{i=fk}^{fk+k-1}$ and $B' = d(u_{fk})$ to $n_{1,1}$, while $n_{1,1}$ then looks at the successors $\{v_i\}_{i=fk}^{fk+k-1}$ and $\{u_{fk}\}$, and add $u_{fk}$ to $D$. This will continue until the condition $|U|<  k2^t$ in the while-loop is violated. This condition gets violated precisely after $n_{0,2^t}$ has been called with $|U| = k2^t$. At this point, we have $B_i' = d(u_{k2^t}) = d(u_{\tau_1})$ and $U =  \{u_i\}_{i=0}^{\tau_1-1} = \cup_{i=1}^{2^t} U_f$, where $U_f = \{u_i\}_{i=fk}^{fk+k-1}$. Since $B > B_i'$, we get that $n_{1,1}$ returns the bound $B' = d(u_{\tau_1})$. Moreover, by Observation 1, $\{x \in W: d(x) < B'\} \subseteq U$ so the set that $n_{1,1}$ returns is $U$. We conclude that $n_{1,1}$ satisfies (iv). Using above, it is also easy to verify that $n_{1,1}$ satisfies (ii) and  (iii). Note that this serves as the base case for (IA2) and (IA3).

    \textbf{Induction step.} We next treat the induction step. That is, consider a node $n_{l,s} \in \tilde{\mathcal{T}}^*$ and assume (IA1) to (IA3) hold. We have three cases: (1) $n_{l,s}$ starts, (2) $n_{l,s}$ calls a node, and (3) $n_{l,s}$ finishes. We consider each case separately.

    \textbf{Case 1: $n_{l,s}$ starts.} Let $n_{l+1,g}$ be the parent of $n_{l,s}$, and assume $l+1<l^*$. We want to show that $n_{l,s}$ satisfies (i) when it starts. 
    
    Assume first that $n_{l,s}$ is the first node that $n_{l+1,g}$ calls. By the induction assumption (IA1), we know that $n_{l+1,g}$ satisfies (i). Therefore,  Observation 1 yields $P = \{u_{g' \tau_{l+1}} \} \cup S_v$ for $n_{l+1,g}$. The elements in $P$ are inserted in $D$ before the while-loop. Hence, $n_{l,s}$ is called (in the while-loop) with input set $S_i \subseteq \{u_{g' \tau_{l+1}} \} \cup S_v $ and bound $B_i \geq T$ since there are only $v$-nodes left  in $D$ or $D$ becomes empty (in which $B_i =B \geq T$ by (IA1)). We have the correct inputs to $n_{l,s}$ provided that $g' \tau_{l+1} = s'\tau_l$. To see this, note that the last $u$-node in the output set $U$ from $n_{l+1,g-1}$, call it $u_j$, and the last $u$-node in the output set $U_i$ from $n_{l,s-1}$, $u_{j'}$, are the same, due to (iv) and (IA3). That is, $j = (g'-1)\tau_{l+1}+\tau_{l+1}$ and $j'=(s'-1)\tau_l+\tau_l$ are equal and therefore $g' \tau_{l+1} = s'\tau_l$. Thus we conclude that $n_{l,s}$ have the correct inputs. What remains to show is that $s' \tau_l+\tau_l<N$. Here, we again use $g' \tau_{l+1} = s'\tau_l$ to get 
    \begin{equation*}
    s' \tau_l+\tau_l = g' \tau_{l+1}+\tau_l < g' \tau_{l+1}+\tau_{l+1} < N 
    \end{equation*}
    using (IA1). Thus, $n_{l,s}$ satisfies (i).

    Assume now that $n_{l,s}$ is not the first node that $n_{l+1,g}$ calls. With an easy induction proof over the previous nodes that $n_{l+1,g}$ has called before $n_{l,s}$ (using (IA2) and (IA3)), we get that $D = \{u_{s'\tau_l}\} \cup S_v'$ just before pulling elements to $n_{l,s}$, where $S_v'$ is a (possibly empty) subset of $v$-nodes. Therefore, similar to the case above, we get that the inputs to $n_{l,s}$ are as in (i). Moreover, by (IA2), we have $s' \leq g'2^t+2^t-1$ and
    \begin{equation*}
    s' \tau_l+\tau_l \leq (g'2^t+2^t-1) \tau_l+\tau_l = g' (2^t \tau_l)+(2^t\tau_l) = g'\tau_{l+1}+\tau_{l+1} < N,
    \end{equation*}
    so $n_{l,s}$ satisfies (i).

    Finally, assume $l+1 = l^*$ so that $n_{l+1,g} = n_{l^*,1}$. Note that the inputs to $n_{l^*,1}$ are as in (i), but that $n_{l^*,1}$ does not satisfy (i) since $\tau_{l^*}\geq N$. However, this will not pose a problem. Indeed, note that $n_{l^*,1}$ will call $n_{l,f}$ with $f = 1,2,\dots,s_{last}$ (in this order) and by definition of $s_{last}$ we have
    \begin{equation*}
    f'\tau_l+\tau_l \leq s_{last}'  \tau_l+\tau_l <N.
    \end{equation*}
    In particular, since $n_{l,s}$ is one of those calls, we have $s'  \tau_l+\tau_l <N$. It remains to show that $n_{l,s}$ receives the correct input. Totally analogous to case $l+1<l^*$ (using an induction proof exploiting (IA3) and the order $n_{l^*,1}$ calls its children), we get $D = \{u_{s'\tau_l}\} \cup S_v'$ just before pulling elements to $n_{l,s}$, where $S_v'$ is a (possibly empty) subset of $v$-nodes. Similarly to case $l+1<l^*$, we conclude that the inputs to $n_{l,s}$ are as in (i). Therefore, $n_{l,s}$ satisfies~(i).

    \textbf{Case 2: $n_{l,s}$ calls a node.} In this case, we want to show that $n_{l,s}$ calls its next child $n_{l-1,f}$ in the correct order given by (ii) using assumption (IA2). 
    
    If $n_{l,s}$ with $l=1$, then we can repeat the derivation in the base case. Indeed, we only need to use that $n_{l,s}$ satisfies (i) to reuse the same derivation, except that we also need to verify that $n_{l,s}$ calls the first node correctly, i.e., $f' = s' 2^t$. Note that $f' = s' 2^t$ is true if and only if $f'\tau_l = s' 2^t \tau_l = s' \tau_{l+1}$, where the latter can be shown analogously to the ``$g' \tau_{l+1} = s'\tau_l$''-derivation in Case 1. We therefore conclude that $n_{l,s}$ calls its next child $n_{l-1,f}$ in the correct order given by (ii). 
    
    Consider now $l\geq 2$. Assume first that $n_{l-1,f}$ is the first node that $n_{l,s}$ calls. Then we want to show $f' = s' 2^t$ which again is true if and only if $f'\tau_l = s' 2^t \tau_l = s' \tau_{l+1}$, where the latter equality can be shown analogously to the ``$g' \tau_{l+1} = s'\tau_l$''-derivation in Case 1. Thus, $n_{l,s}$ calls $n_{l-1,f}$ correctly.

    Assume next that $n_{l-1,f}$ is not the first node that $n_{l,s}$ calls. To show that $n_{l,s}$ calls $n_{l-1,f}$ correctly, it suffices to show that $n_{l,s}$ will exit the while-loop after $n_{l-1,f}$ has returned if $f'=s'2^t+2^t-1$ and continue if $f'<s'2^t+2^t-1$. To this end, assume $n_{l-1,f}$ returns with $f'=s'2^t+j-1$ for $j \leq 2^t$. At this point, due to (IA2) and (IA3), we have 
    \begin{equation*}
    |U| = \sum_{i=1}^j |U_i| = j \cdot k2^{(l-1)t} \leq 2^t k2^{(l-1)t} = k2^{lt}
    \end{equation*}
    Since $D$ never becomes empty (since $n_{l,s}$ satisfies $s'\tau_l+\tau_l< N$), the only way we can exit the while-loop is if $|U| \geq k2^{lt}$. By above, this happens when $j=2^t$ (otherwise the while-loop continues if $j<2^t$). We conclude that $n_{l,s}$ calls $n_{l-1,f}$ correctly.

    \textbf{Case 3: $n_{l,s}$ finishes.} In this case, we want to show that $n_{l,s}$ satisfies (iii) and (iv). By (IA2), the output set $U_i$ in the while-loop of $n_{l,s}$ comes from $n_{l-1,f}$ with $f'=s'2^t+i-1$. Moreover, by (IA3), we have that $U_i = \{u_{f'\tau_{l-1}+j}\}_{j=0}^{\tau_{l-1}-1}$. It is easy to see that all these $U_i$:s (with $i=1,\dots,2^t$) form a partition of consecutive $u$-nodes with the first $u$-node given by
    \begin{equation*}
    u_{s'2^t\tau_{l-1}}
    \end{equation*}
    and the last $u$-node given by 
    \begin{equation*}
    u_{(s'2^t+2^t-1)\tau_{l-1}+\tau_{l-1}-1}.
    \end{equation*}
    Noting that $s'2^t\tau_{l-1} = s' \tau_l$ and $(s'2^t+2^t-1)\tau_{l-1}+\tau_{l-1}-1 = s'\tau_l+\tau_l-1$ we get that $\cup_{i=1}^{2^t} U_i = \{u_{s'\tau_l+i}\}_{i=0}^{\tau_l-1}$. When $n_{l,s}$ exits its while-loop, we have $U= \cup_{i=1}^{2^t} U_i = \{u_{s'\tau_l+i}\}_{i=0}^{\tau_l-1}$. By an easy proof of induction over the function calls $n_{l-1,f}$, we get that $D$ is equal to
    \begin{equation*}
    D = \{ u_{s'\tau_l+\tau_l}\} \cup S_v'
    \end{equation*}
    when $n_{l,s}$ exits its while-loop, where $S_v'$ is a (potentially empty) set of $v$-nodes. Thus $B_i' = d(u_{s'\tau_l+\tau_l})$ when the while-loop is exited and since $B\geq T > B_i'$, we get that $n_{l,s}$ returns the bound $B' = B_i' = d(u_{s'\tau_l+\tau_l})$. Moreover, $U$ does not change after the while-loop since $\{x \in W: d(x)<B'\} \subseteq \{u_{s'\tau_l+i}\}_{i=0}^{\tau_l-1} = U$ (due to (IA1) and Observation 1). We conclude that $n_{l,s}$ returns the set $U=\{u_{s'\tau_l+i}\}_{i=0}^{\tau_l-1}$, where $U= \cup_{i=1}^{2^t} U_i$ is a partition of $U$. Thus, $n_{l,s}$ satisfies (iv).

    To show (iii), simply note that the nodes $\{v_{s'\tau_l+i}\}_{i=0}^{\tau_l-1}$ are successors of $U= \cup_{i=1}^{2^t} U_i$. Thus, $n_{l,s}$ must have looked at each node in $\{v_{s'\tau_l+i}\}_{i=0}^{\tau_l-1}$ in the while-loop. Thus, $n_{l,s}$ satisfies (iii).

    We have completed Cases 1 to 3 in the induction step. Therefore, by induction, Lemma \ref{lemma:superlinear2} follows.
\end{proof}

Using Lemma \ref{lemma:superlinear2}, we can show the superlinear time complexity, given by the following~corollary.

\begin{corol}
Running the shortest path algorithm from \cite{duan2025breaking} (Algorithm 3) on the graph $G$ results in superlinear time complexity with lower bound $\Omega(\frac{N}{2}l_{top})=\Omega(n \log^{1/3}(n))$, originating from the number of times we look at $v$-nodes.
\end{corol}

\begin{proof}
    Consider $n_{l^*,1}$ and the calls $n_{l^*-1,f}$ with $f=1,2,\dots,s_{last}$. Using Lemma \ref{lemma:superlinear2}, is easy to see that the number of $u$-nodes that the calls $n_{l^*-1,f}$ go through must be at least the first half of the total number of $u$-nodes (due of the definitions of $s_{last}$ and $l^*$). Let these $u$-nodes be $\{u_i\}_{i=0}^{L-1}$ so that $L \geq \frac{N}{2}$. Again, by Lemma \ref{lemma:superlinear2}, these nodes $\{u_i\}_{i=0}^{L-1}$ are in a $U_i$ at each level $l=0,1,\dots,l^*-1$. Therefore, the successors $\{v_i\}_{i=0}^{L-1}$ of $\{u_i\}_{i=0}^{L-1}$ will be looked at in each level $l=1,2,\dots,l^*-1$. Hence, the total number of looks on $v$-nodes is lower bounded by
    \begin{equation*}
        \Omega(L l^*) \geq \Omega \left (\frac{N}{2} l^*\right ).
    \end{equation*}
    Next we show that $l^* \in \Theta(l_{top})$. To see this, note that $2^{l_{top}t} = \Theta(n) = \Theta(N)$. In particular, $2^{l_{top}t} \leq c N \leq ck2^{l^* t}$ for some constant $c>1$. Thus,
\begin{equation*}
l_{top}t \leq  \log(c)+\log(k)+l^* t
\end{equation*}
and so
\begin{equation*}
l_{top} \leq  \frac{\log(c)+\log(k)}{t}+l^*.
\end{equation*}
Since $\frac{\log(c)+\log(k)}{t} \rightarrow 0$ as $n \rightarrow \infty$, we get that $l_{top} \leq  C+l^*$ for some constant $C>0$ and thus $l^* \in \Omega(l_{top})$. Since $l^* \leq l_{top}$, we have $l^* \in \Theta(l_{top})$. Therefore, 
\begin{equation*}
\Omega\left (\frac{N}{2}l^* \right) = \Omega\left (\frac{N}{2}l_{top} \right) = \Omega\left ( n \log^{1/3}(n) \right)
\end{equation*}
is a lower bound on the time complexity. This completes the proof.
\end{proof}

\end{document}

%% file: preamble.tex
\usepackage[utf8]{inputenc}
\usepackage{hyperref,url}

\usepackage[dvipsnames]{xcolor} 

\usepackage{geometry}
\usepackage{graphicx}
\usepackage{adjustbox}

\usepackage{booktabs}
\usepackage{array}
\usepackage{paralist}
\usepackage{verbatim}
\usepackage{caption}
\usepackage{subcaption}
\usepackage{amsmath}
\usepackage{amsthm}
\usepackage{amssymb}
\usepackage{mathdots}
\usepackage[linesnumbered,ruled,vlined,noend]{algorithm2e}
\usepackage{todonotes}
\usepackage{bm}
\usepackage{xargs}
\usepackage{comment}
\usepackage{standalone}
\usepackage{tikz}
\usepackage{tikz-cd}
\usepackage{tikz-network}
\usepackage{enumitem}

\usetikzlibrary{decorations.pathreplacing}
\usetikzlibrary{decorations.pathmorphing}
\usetikzlibrary{automata}
\usetikzlibrary{arrows.meta}
\tikzset{every state/.style={inner sep=1pt,minimum size=14pt}}
\tikzset{edge/.style = {->, semithick, >=Latex}}
\tikzset{edgelabel/.style = {inner sep=1pt, fill=white}}
\tikzset{init/.style = {initial, initial text={}}}
\tikzset{>=Latex}
\tikzset{ball/.style={circle,draw,fill=black,inner sep=0pt, minimum width=4pt}}

\theoremstyle{plain}
\newtheorem{thm}{Theorem}[section]
\theoremstyle{definition}
\newtheorem{defn}[thm]{Definition}
\newtheorem{lem}[thm]{Lemma}
\newtheorem{corol}[thm]{Corollary}
\newtheorem{rmk}{Remark}

\newtheorem{prop}[thm]{Proposition}

\newcommand{\real}{\mathbb{R}}

\newcommand{\pth}[3][]{#2 \stackrel{#1}{\leadsto} #3}
\newcommand{\arc}[3][]{#2 \xrightarrow{#1} #3}

\newcommand{\overallalg}{\texttt{A-C\_decomposition}}


%% file: figs/intro_ex3.tex
\begin{tikzpicture}
\node (1) [state,initial below, initial text={}] at (5,3) {$s$};
\node (2) [state] at (4,0) {$a$};
\node (3) [state] at (4,1) {$h$};
\node (4) [state] at (3,4) {$b$};
\node (5) [state] at (1,4) {$c$};
\node (6) [state] at (2,2) {$d$};
\node (7) [state] at (2,0) {$e$};
\node (8) [state] at (0,2) {$f$};
\node (9) [state] at (0,0) {$g$};
\node[draw,color=gray,rounded corners=3.5mm,minimum height=2.8cm, minimum width=4.8cm,line width=0.5mm] (group) at (2,1) {};
\draw[rounded corners=3.5mm,draw=white,line width=0.5mm] (-1,-0.75) -- (-1,2.5) -- (0.25,4.75) -- (3,4.75) -- (5.5,3.5) -- (5.5,-0.75) -- cycle;
\draw[edge] (1) to (4);
\draw[edge] (3) to (1);
\draw[edge] (3) to (2);
\draw[edge] (3) to (6);
\draw[edge] (3) to (7);
\draw[edge] (4) to (3);
\draw[edge] (4) to[out=150, in=30] (5);
\draw[edge] (5) to[out=150+180, in=30+180] (4);
\draw[edge] (6) to[out=150+90, in=30+90] (7);
\draw[edge] (6) to (8);
\draw[edge] (7) to[out=150+90+180, in=30+90+180] (6);
\draw[edge] (8) to (5);
\draw[edge] (8) to (9);
\draw[edge] (9) to (6);

\draw[edge] (2) to (7);
 \end{tikzpicture}

%% file: figs/intro_ex3_nested.tex
\begin{tikzpicture}
\node (1) [state,initial below, initial text={}] at (5,3) {$s$};
\node (2) [state] at (4,0) {$a$};
\node (3) [state,initial above, initial text={}] at (4,1) {$h$};
\node (4) [state] at (3,4) {$b$};
\node (5) [state] at (1,4) {$c$};
\node (6) [state] at (2,2) {$d$};
\node (7) [state] at (2,0) {$e$};
\node (8) [state] at (0,2) {$f$};
\node (9) [state] at (0,0) {$g$};
\node[draw,color=gray,rounded corners=3.5mm,minimum height=2.8cm, minimum width=4.8cm,line width=0.5mm] (group) at (2,1) {};
\node[] at (2,2.75) {$X$};
\node[] at (3.4,2) {$G_2$};
\draw[draw,color=gray,rounded corners=3.5mm,line width=0.5mm] (-1,-0.75) -- (-1,2.5) -- (0.25,4.75) -- (3,4.75) -- (5.5,3.5) -- (5.5,-0.75) -- cycle;
\node[] at (5,4.125) {$G_1$};
\draw[edge] (1) to (4);
\draw[edge] (group) to (1);
\draw[edge] (3) to (2);
\draw[edge] (3) to (6);
\draw[edge] (3) to (7);
\draw[edge] (4) to (group);
\draw[edge] (4) to[out=150, in=30] (5);
\draw[edge] (5) to[out=150+180, in=30+180] (4);
\draw[edge] (6) to[out=150+90, in=30+90] (7);
\draw[edge] (6) to (8);
\draw[edge] (7) to[out=150+90+180, in=30+90+180] (6);
\draw[edge] (group) to (5);
\draw[edge] (8) to (9);
\draw[edge] (9) to (6);

\draw[edge] (2) to (7);
 \end{tikzpicture}

%% file: figs/intro_plain.tex
\begin{tikzpicture}
\node (1) [state,init] at (5,3) {$s$};
\node (2) [state] at (4,0) {$a$};
\node (3) [state] at (4,1) {$h$};
\node (4) [state] at (3,4) {$b$};
\node (5) [state] at (1,4) {$c$};
\node (6) [state] at (2,2) {$d$};
\node (7) [state] at (2,0) {$e$};
\node (8) [state] at (0,2) {$f$};
\node (9) [state] at (0,0) {$g$};

\draw[edge] (2) to (7);
\draw[edge] (1) to (4);
\draw[edge] (3) to (1);
\draw[edge] (3) to (2);
\draw[edge] (3) to (6);
\draw[edge] (3) to (7);
\draw[edge] (4) to (3);
\draw[edge] (4) to[out=150, in=30] (5);
\draw[edge] (5) to[out=150+180, in=30+180] (4);
\draw[edge] (6) to[out=150+90, in=30+90] (7);
\draw[edge] (6) to (8);
\draw[edge] (7) to[out=150+90+180, in=30+90+180] (6);
\draw[edge] (8) to (5);
\draw[edge] (8) to (9);
\draw[edge] (9) to (6);
 \end{tikzpicture}

%% file: figs/AC_tree_new.tex
\begin{tikzpicture}
\tikzset{every state/.style={inner sep=1pt,minimum size=16pt}}

\node[draw,color=gray,shape=rectangle,minimum width=1.9cm,minimum height=.9cm,rounded corners] (rect) at (1.5,-4) {};
\node[line width=0.5mm] (s) [state] at (2,-1) {$s$};
\node[draw,line width=0.5mm] (a) [state] at (0,-4) {$a$};
\node[line width=0.5mm] (b) [state] at (2,-2) {$b$};
\node[draw,line width=0.5mm] (c) [state] at (3,-3) {$c$};
\node[draw,line width=0.5mm] (r) [state] at (1,-3) {$h$};
\node[draw,line width=0.5mm] (e) [state] at (1,-4) {$e$};
\node[line width=0.5mm] (d) [state] at (2,-4) {$d$};
\node[draw,line width=0.5mm] (f) [state] at (2,-5) {$f$};
\node[draw,line width=0.5mm] (g) [state] at (2,-6) {$g$};

\draw[->] (s) to (b);
\draw[->] (b) to (c);
\draw[->] (b) to (r);
\draw[->] (r) to (a);
\draw[->] (r) to (rect);
\draw[->] (d) to (f);
\draw[->] (f) to (g);
 \end{tikzpicture}